\documentclass[reqno]{amsart}

\usepackage[linesnumbered,ruled]{algorithm2e}
\SetKwInput{KwInput}{Input}                
\SetKwInput{KwOutput}{Output}              

\usepackage{booktabs} 
\usepackage{IEEEtrantools}
\usepackage{amssymb,latexsym,amsfonts,amsmath}
\usepackage{dsfont}

\usepackage{extarrows}

 \usepackage{amssymb,latexsym,amsfonts,amsmath}
\usepackage{graphicx}
\usepackage{paralist}
\usepackage{capt-of}
\usepackage{xcolor}
\usepackage{dsfont}
\usepackage{tikz}

\newcommand{\tup}{\textup}

\def\boxspan{\mathit{span}}

\usetikzlibrary{calc,shapes,arrows,intersections,automata,shapes,chains}

\newcommand{\R}{{\mathbb{R}}}

\newcommand{\N}{{\mathbb{N}}}

\newcommand{\eg}{{\it e.g.}}

\newcommand{\KK}{\mathcal{K}_{\infty}}

\newcommand{\Let}{:=}

\newcommand{\intcc}[1]{\ensuremath{{\left[#1\right]}}}

\definecolor{myco}{rgb}{0.55, 0.0, 0.63}

\newtheorem{theorem}{Theorem}[section]

\newtheorem{proposition}[theorem]{Proposition}

\newtheorem{definition}[theorem]{Definition}

\newtheorem{remark}[theorem]{Remark}
\newtheorem{assumption}[theorem]{Assumption}
\numberwithin{equation}{section}

\makeatletter
\long\def\@maketablecaption#1#2{\@tablecaptionsize
	\global \@minipagefalse
	\hbox to \hsize{\parbox[t]{\hsize}{\centering #1 \\ #2}}}
\makeatother

\begin{document}
	
	\begin{abstract}
 In this paper, we propose a compositional approach to construct opacity-preserving finite abstractions (a.k.a symbolic models) for networks of discrete-time nonlinear control systems. Particularly, we introduce new notions of simulation functions that characterize the distance between control systems while preserving opacity properties across them. Instead of treating large-scale systems in a monolithic manner, we develop a compositional scheme to construct the interconnected finite abstractions together with the overall opacity-preserving simulation functions. For a network of incrementally input-to-state stable control systems and under some small-gain type condition, an algorithm for designing local quantization parameters is presented to orderly build the local symbolic models of subsystems 
 such that the network of symbolic models simulates the original network for an a-priori defined accuracy while preserving its opacity properties.
	\end{abstract}
	
	\title[Compositional synthesis of opacity-preserving finite abstractions for interconnected systems]{Compositional synthesis of opacity-preserving finite abstractions for interconnected systems}

	\author{Siyuan Liu$^1$}
	\author{Majid Zamani$^{2,3}$}
	\address{$^1$Electrical and Computer Engineering Department, Technical University of Munich, Germany.}
	\email{sy.liu@tum.de}
	\address{$^2$Computer Science Department, University of Colorado Boulder, USA.}
	\email{majid.zamani@colorado.edu}
	\address{$^3$Computer Science Department, Ludwig Maximilian University of Munich, Germany.}
	\maketitle
	
	\section{Introduction}
	In the recent decade, the world has witnessed a rapid increase in applications of cyber-physical systems (CPSs), which are networked systems resulting from intricate interactions of   cyber components and physical plants. CPSs play a major role in our daily life and many safety-critical infrastructure, such as autonomous vehicles, implantable and wearable medical devices and smart communities. On the other hand, new threats have been continuously affecting the performance and safety of such applications. One of the major issues is security problems. In particular, the complex interaction between embedded (cyber) software and physical devices may release secret information and expose the system to (cyber) attackers. Therefore, new approaches to analyze or enforce security properties over safety-critical CPSs have drawn significant attentions in the past few years \cite{ashibani2017cyber,sandberg2015cyberphysical}. 
	
	In this paper, we focus on an information-flow security property called \emph{opacity}, which was originally proposed in the realm of computer science for the analysis of cryptographic protocols \cite{mazare2004using} but has not been thoroughly investigated in the domain of CPS. As a confidentiality property, opacity characterizes the ability of a system to avoid leaking ``secret'' information in the presence of outside observers with potentially malicious intentions. Intuitively, a system is called \emph{opaque} if it has the plausible deniability for its ``secret'' so that the outside observers cannot infer or determine the system's secret based on its (partial) observations. In discrete-event systems (DESs) literature, different notions of opacity have been proposed in order to capture various types of secret requirements, including state-based notions in \cite{saboori2011verification,saboori2012verification,saboori2013verification} and language-based notions in \cite{lin2011opacity}. In practical situations, the state-based notions of opacity of DESs are generally classified into the so-called initial-state opacity \cite{saboori2013verification}, current-state opacity \cite{saboori2007notions}, K-step opacity \cite{saboori2011verification}, and infinite-step opacity \cite{saboori2012verification}. Later on, more research on opacity for various classes of discrete systems has been conducted \cite{tong2017decidability, saboori2013current, chedor2015diagnosis}. We refer interested readers to the recent surveys in  \cite{jacob2016overview,lafortune2018history} for more  details about opacity of DESs. 
	 	
	Unfortunately, most of the existing results on opacity are tailored to DESs, where they consider the event-based observation model, i.e., some events of the system are observable or distinguishable
while some are not. Whereas in real-world applications, outputs are typically physical signals equipped with some metrics and state space are usually continuous. To this purpose, in some recent works \cite{8264445,ramasubramanian2019notions}, the notion of opacity was extended to discrete-time (switched) linear systems.  However, their definition of opacity is more related to an output reachability
property rather than an information-flow one. 
To the best of our knowledge, most of the existing results on opacity are not suitable for capturing the information-flow security for real-world CPSs.

In this work, we aim at leveraging symbolic techniques to tackle this property for CPSs. In particular, we address this property by constructing finite abstractions (a.k.a symbolic models) of the concrete systems based on some types of opacity-preserving simulation relations between the concrete systems and their abstractions. These relations enable us to verify or enforce opacity for the concrete systems by performing the corresponding analysis over the simpler finite ones. 
Moreover, by following such a detour process, one can leverage (by some adaptation) existing computational tools developed in the DESs literature to verify or enforce opacity over CPSs.

In recent years, there have been some attempts in the literature to leverage abstraction-based techniques for the verification or enforcement of opacity \cite{wu2018privacy,zhang2019opacity,yin2019approximate,liu2020stochastic}. 
The result in \cite{wu2018privacy} introduced an abstract model based on the belief space of the intruder, using which controllers are synthesized to enforce opacity. However, the systems considered there is modeled as transition systems with finite state set, thus, not suitable for general CPSs. 
In \cite{zhang2019opacity}, a new formulation of opacity-preserving (bi)simulation relations is proposed, which allows one to verify opacity of an infinite-state transition system by leveraging its associated quotient one. However, the notion of opacity proposed in this work assumes that the outputs of systems are symbols and exactly distinguishable from each other, thus, is only suitable for systems with purely logical output sets. 
In \cite{yin2019approximate}, a new notion called \emph{approximate opacity} is proposed to suitably capture the continuity of output spaces of real-world CPSs. Additionally, a new simulation relation, called approximate opacity-preserving simulation relation, was proposed to characterize the closeness of two (finite or infinite) systems while preserving approximate opacity across them. The recent results in \cite{liu2020stochastic} investigate opacity for discrete-time stochastic control systems using a notion of so-called initial-state opacity-preserving stochastic simulation functions between stochastic control systems and their finite abstractions in the form of finite Markov Decision Processes (MDPs). Though promising, when confronted with large-scale interconnected systems, the construction of finite abstractions in the aforementioned works will suffer severely from the \emph{curse of dimensionality} because the number of discrete states grows exponentially with the dimension of the concrete state set.
	
	Motivated by the abstraction-based techniques in \cite{zhang2019opacity,yin2019approximate,liu2020stochastic} and their computational complexity, here, we aim at providing a compositional framework to conquer this complexity challenge using a ``divide and conquer" strategy. To this purpose, we first introduce new notions of opacity-preserving simulation functions for both local subsystems and the entire networks. Based on this, we propose a compositional scheme on the construction of abstractions for concrete networks. Rather than dealing with the original large-scale system, our compositional framework allows one to construct opacity-preserving abstractions locally using local opacity-preserving simulation functions, while providing the guarantee that the interconnection of local finite abstractions simulates the concrete network while preserving opacity across them.
	
	First, by considering three basic notions of opacity, namely initial-state opacity, current-state opacity, and infinite-step opacity, we introduce new notions of initial-state (resp. current-state, and infinite-step) opacity-preserving simulation functions. Given the opacity-preserving simulation functions between subsystems and their finite abstractions, we provide a compositionality result showing that the interconnection of finite abstractions retains an opacity-preserving simulation relation with the original network. The overall opacity-preserving simulation function is constructed compositionally from the local simulation functions. Additionally, considering the class of incrementally input-to-state stable control systems, we exploit the interconnection topology of the network and present an algorithm to design quantization parameters of the local finite abstractions together with the local opacity-preserving simulation functions. Finally, we illustrate the effectiveness of our results on some examples.
	
Remark that compositional approaches have been also investigated recently for controller synthesis of interconnected CPSs, see e.g., \cite{tazaki2008bisimilar,pola2016symbolic,swikir2019compositional, kim2018constructing, mallik2018compositional}. 
	Unfortunately, none of those techniques is applicable to the verification or enforcement of opacity mainly because their underlying system relations do not necessarily preserve opacity across the related systems.

    \section{Notation and Preliminaries}
	\subsection{Notation}
	We denote by $\R$ and $\N$ the set of real numbers and non-negative integers, respectively.
	These symbols are annotated with subscripts to restrict them in
	the usual way, \eg, $\R_{>0}$ denotes the positive real numbers. We denote the closed, open, and half-open intervals in $\R$ by $[a~b]$,
	$]a~b[$, $[a~b[$, and $]a~b]$, respectively. For $a,b\!\in\!\N$ and $a\!\le\! b$, we
	use $[a;b]$, $]a;b[$, $[a;b[$, and $]a;b]$ to
	denote the corresponding intervals in $\N$.	
	Given $N \!\in \!\mathbb N_{\ge 1}$ vectors $x_i \!\in\! \mathbb R^{n_i}$, with $i\!\in\! [1;N]$, $n_i\!\in\! \mathbb N_{\ge 1}$, and $n \!= \!\sum_i n_i$, we denote the concatenated vector  in $\mathbb R^{n}$ by $x \!=\! [x_1;\!\ldots\!;x_N]$  and the infinity norm of $x$ by $\Vert x\Vert$.  
 We denote by $\{A\}_{ij}$ the individual elements in a matrix $A\!\in\! \R^{m\!\times\!n}\!$ and by $0_n$ the zero matrix in $\mathbb{R}^{n \!\times \!n}$. We use $\textup{card}(\cdot)$ to denote the cardinality of a set and $\varnothing$ to denote the empty set.  Given any $a\!\in\!\R$, $\vert a\vert$ denotes the absolute value of $a$. 
The composition of functions $f$ and $g$ is denoted by $f \!\circ\! g$.
We use notations $\mathcal{K}$ and $\mathcal{K}_{\infty}$ to denote the different classes of comparison functions, as follows: $\!\mathcal{K} \!\!=\! \!\{\gamma\!:\!\mathbb R_{\ge 0}\!\!\rightarrow\!\!\mathbb R_{\ge 0}  \! \mid  \! \gamma \text{ is continuous, strictly increasing and } \gamma(0)\!=\!0\}$; $\!\mathcal{K}_{\infty} \!\!=\! \!\{\gamma \!\in \!\mathcal{K} \! \!\mid \!  \lim\limits_{r\rightarrow \infty}\!\!\gamma(r) \!=\!\infty\}$.
For $\alpha$,$\gamma \!\in \!\mathcal{K}_{\infty}$ we write $\alpha\!<\!\gamma$ if $\alpha(s)\!<\!\gamma(s)$ for all $s\!>\!0$, and $\mathcal{I}_d\!\in\!\mathcal{K}_{\infty}$ denotes the identity function. 
Given sets $X$ and $Y$ with $X \! \subset  \! Y$, the complement of $X$ with respect to $Y$ is defined as $Y \backslash X \!=\! \{x :\! x \!\in\! Y, x \!\notin\! X\}.$

	The closed ball centered at $u\in{\mathbb{R}}^{m}$ with radius $\lambda$ is defined by $\mathcal{B}_{\lambda}(u)\!=\!\{v\in{\mathbb{R}}^{m}\,|\,\Vert u\!-\!v\Vert\!\leq\!\lambda\}$. We denote the closed ball centered at the origin in $\R^n$ and with radius $\lambda$ by $\mathcal{B}_{\lambda}$.
Consider a set $A$ of the form of finite union of boxes, i.e. $A \!=\! \bigcup_{j=1}^M A_j$, where $A_j\!=\! \prod_{i=1}^m [c_i^j, d_i^j] \!\subseteq\! \R^m$ with $c_i^j\!<\! d_i^j$. Define $\boxspan(A_j)\! \!=\! \!\min\{ | d_i^j \!-\! c_i^j| \!\mid\! i\!\in\!\intcc{1;m}\}$ and $\boxspan(A) \!\!=\!\! \min\{\boxspan(A_j)\!\mid\! j\!\in\!\intcc{1;M}\}$. For any quantization parameter $\eta$ with $\eta \!\leq\! \boxspan(A)$, define $[A]_\eta\!=\! \bigcup_{j=1}^M [A_j]_{\eta}$, where $[A_j]_{\eta} \!\!= \![\R^m]_{\eta}\!\cap\!{A_j}$ with $[\R^m]_{\eta}\!=\!\{a\!\in\! \R^m\!\mid\! a_{i}\!=\!k_{i}\eta,k_{i}\!\in\!\mathbb{Z},i\!\in\!\intcc{1;m}\}$. 
Note that $[A]_{\eta}\!\neq\!\varnothing$ for any $0\!\leq\!\eta\!\leq\!\boxspan(A)$. 
	With a slight abuse of notation, we write $[A]_0 \!:=\!A$. 
	
	For any set $A=\prod_{j=1}^N A_j$, where $A_j$ are of the form of finite union of boxes, and a vector of quantization parameters $\eta\! =\! [\eta_1;\!\dots\!;\eta_N]$ with $\eta_j \!\leq\! \boxspan(A_j)$, $\forall j\!\in\! [1;M]$, define $[A]_\eta\!=\! \prod_{j=1}^N [A_j]_{\eta_j}$. 
	Note that if $\eta\! =\! [\mu;\!\dots\!;\mu]$, we simply use notation $[A]_\mu$ rather than $[A]_\eta$. 
	The Minkowski sum of two sets $P,Q \!\subseteq\! \R^n$ is defined by $P \!\oplus Q \!= \{x \!\in\! \R^n|\exists_{p \in P, q \in Q}, x \!=\! p \!+\! q\}$. Given a set $\mathbb{S}\!\subseteq\!\R^n$ and a constant $\theta \!\in\! \mathbb R_{\ge 0}$, we define a new set $\mathbb S^{\theta}\! \!= \! \mathbb{S} \!\oplus\! \mathcal{B}_{\theta}$ as the inflated version of set $\mathbb{S}$.
	
	A \textit{directed graph} is denoted by $G \!=\! (\mathcal{V}, \mathcal{E})$, where $\mathcal{V}$ is the set of vertices, and $\mathcal{E} \!\subseteq\! \mathcal{V} \!\times \!\mathcal{V}$ is the set of edges with direction, where a directed edge is denoted by an ordered pair $(i,j)$ of vertices, $\forall i,j \!\in \!\mathcal{V}$, if there is an incoming edge from $j$ to $i$. 	A directed graph with no directed cycles is called \textit{acyclic}.  


	\subsection{Discrete-time control systems} 
	In this paper we study the class of discrete-time control systems of the following form.
	\begin{definition}\label{def:sys1}
		A discrete-time control system $\Sigma$ is defined by the tuple	$\Sigma\!=(\mathbb X,\mathbb U,\mathbb W,\mathcal{U},\mathcal{W},f,\mathbb Y,h)$
		where $\mathbb X$, $\mathbb U$, $\mathbb W$ and $\mathbb Y$ are the state set, external input set, internal input set, and output set, respectively. Sets $\mathcal{U}$ and $\mathcal{W}$, respectively, are used to denote the subsets of the set of all bounded functions $\nu:\N\rightarrow \mathbb U$ and $\omega:\N\rightarrow \mathbb W$, respectively. The set-valued map $f: \mathbb X\times \mathbb U \times \mathbb W\rightrightarrows \mathbb X $ is the state transition function, and $h:\mathbb X \rightarrow \mathbb Y$  is the output function.
		The discrete-time control system $\Sigma $ is described by difference inclusions of the form
		\begin{align}\label{eq:2}
		\Sigma:\left\{
		\begin{array}{rl}
		\mathbf{x}(t+1)\in& f(\mathbf{x}(t),\nu(t),\omega(t)),\\
		\mathbf{y}(t)=&h(\mathbf{x}(t)),
		\end{array}
		\right.
		\end{align}
		where $\mathbf{x}:\mathbb{N}\rightarrow \mathbb X $, $\mathbf{y}:\mathbb{N}\rightarrow \mathbb Y$, $\nu\in\mathcal{U}$, and $\omega\in\mathcal{W}$ are the state, output, external input, and internal input signals, respectively. We assume the output set $Y$ is equipped with the infinity norm as the metric defined on this set.
		
		System $\Sigma\!=(\mathbb X,\mathbb U,\mathbb W,\mathcal{U},\mathcal{W},f,\mathbb Y,h)$ is called deterministic if $\textup{card}(f(x,u,w))\leq1$ $ \forall x\in \mathbb X, \forall u\in \mathbb U, \forall w \in \mathbb W$, and non-deterministic otherwise. System $\Sigma$ is called blocking if $\exists x\in \mathbb X, \forall u\in \mathbb U, \forall w \in \mathbb W $ where $\textup{card}(f(x,u,w))=0$ and non-blocking if $\textup{card}(f(x,u,w))\neq 0$ $ \forall x\in \mathbb X, \exists u\in \mathbb U, \exists w \in \mathbb W$.  System $\Sigma$ is called finite if $\mathbb X,\mathbb U,\mathbb W$ are finite sets and infinite otherwise. In this paper, we only deal with non-blocking systems. 
	\end{definition}	
	Note that in our paper, we always consider systems with secret states which are supposed to be hidden from the intruder. Hereafter, we slightly modify the formulation in Definition \ref{def:sys1} to accommodate for sets of initial and secret states, as $\Sigma\!=(\mathbb X,\mathbb X_0,\mathbb X_S,\mathbb U,\mathbb W,\mathcal{U},\mathcal{W},f,\mathbb Y,h)$, where $\mathbb X_0 \subseteq \mathbb X$ is a set of initial states and $\mathbb X_S \subseteq \mathbb X$ is a set of secret states.

	Now, we provide a formal definition of interconnected control systems.
	Consider $N \in \N_{\geq 1}$ control subsystems 
	\begin{align}
	\label{eq:subcontrolsystem}
    \Sigma_i\!=(\mathbb X_i,\mathbb X_{i0},\mathbb X_{iS},\mathbb U_i,\mathbb W_i,\mathcal{U}_i,\mathcal{W}_i,f,\mathbb Y_i,h_i),
	\end{align}
	where $i \in [1;N]$, and assume internal inputs and output maps are partitioned as
	\begin{IEEEeqnarray}{c}	
		\label{internalinput}	
		w_i=[w_{i1};\ldots;w_{i(i-1)};w_{i(i+1)};\ldots;w_{iN}],\\
		\label{output}	
		h_i(x_i) = [h_{i1}(x_i);\dots;h_{iN}(x_i)],
	\end{IEEEeqnarray}
	with $\mathbb W_i=\prod_{j=1, j\neq i}^{N} \mathbb W_{ij}$ and $\mathbb Y_i=\prod_{j=1}^N \mathbb Y_{ij}$, $w_{ij} \!\in\! \mathbb{W}_{ij}$, $y_{ij} \!=\!h_{ij}(x_i)\!\in \!\mathbb Y_{ij}$.

	The outputs $y_{ii}$ are considered as external ones, whereas $y_{ij}$ with $i \neq j$ are interpreted as internal ones which are used to construct interconnections between subsystems. The dimension of $w_{ij}$ is assumed to be equal to that of $y_{ji}$. In the case that no connection exists between subsystems $\Sigma_i$ and $\Sigma_j$, we simply have $h_{ij} \equiv 0$. The interconnected control system is defined as the following.
	\begin{definition}
		\label{interconnectedsystem} 
		Consider $N \in \N_{\geq 1}$ control subsystems $\Sigma_i\!=(\mathbb X_i,\mathbb X_{i0},\mathbb X_{iS},\mathbb U_i,\mathbb W_i,\mathcal{U}_i,\mathcal{W}_i,f,\mathbb Y_i,h_i)$, $i\in[1;N]$, with the input-output structure given in \eqref{internalinput}-\eqref{output}. The interconnected control system denoted by $\mathcal{I}_{\mathcal{M}}(\Sigma_1,\dots,\Sigma_N)$ is a tuple
		\begin{align}
		\label{interConnectedsys}
		\Sigma\!=(\mathbb X,\mathbb X_0,\mathbb X_S,\mathbb U,\mathcal{U},f,\mathbb Y,h),
		\end{align} 
		where $ \mathbb X =\prod_{i=1}^N \mathbb X_i$, $\mathbb X_0 =\prod_{i=1}^N \mathbb X_{i0}$, $\mathbb X_S =\prod_{i=1}^N \mathbb X_{iS}$, 		
		$ \mathbb U=\prod_{i=1}^N \mathbb U_i$, $\mathcal{U}=\prod_{i=1}^N\mathcal{U}_i$, $ \mathbb Y=\prod_{i=1}^N \mathbb Y_{ii}$, the state transition and output functions are   
		\begin{align}\notag
		f(x,u)&=\{[x_{1}';\dots;x_{N}']|x_{i}' \in f_i(x_i,u_i,w_i), \forall i \in [1;N]\},\\ \notag
		h(x) &= [h_{11}(x_1);\dots;h_{NN}(x_N)],
		\end{align} 
		where $x = [x_1;\dots;x_N]$, $u = [u_1;\dots;u_N]$, $\mathcal{M} \in \mathbb{R}^{N \times N}$ is a matrix with elements $\{\mathcal{M}\}_{ii} = 0,\{\mathcal{M}\}_{ij} = \phi_{ij}, \forall i,j \in [1;N], i\neq j $, $0\leq\phi_{ij}\leq \boxspan(\mathbb{Y}_{ji})$, and the interconnection variables are constrained by
		\begin{align}\label{interconstraint}
        \Vert y_{ji}-w_{ij} \Vert \leq \phi_{ij},\quad [\mathbb{Y}_{ji}]_{\phi_{ij}} \subseteq \mathbb{W}_{ij}, \quad  \forall i,j \in [1;N], i\neq j. 
		\end{align}	
	 The set-valued map $f$ becomes $f:\mathbb X\times\mathbb U\rightrightarrows\mathbb X$ and \eqref{eq:2} reduces to
\begin{align}\label{eq:3}
\Sigma:\left\{
\begin{array}{rl}
\mathbf{x}(t+1)\in&f(\mathbf{x}(t),\nu(t)),\\
\mathbf{y}(t)=&h(\mathbf{x}(t)).
\end{array}\right.
\end{align}

\end{definition} 

\begin{remark}
Note that condition \eqref{interconstraint} is required to provide a well-posed interconnection. 
		Throughout this paper, when we are talking about the concrete interconnected system, $y_{ji}$ is always equal to $w_{ij}$ (i.e. $y_{ji}=w_{ij}$), which naturally implies $\phi_{ij}=0$ and $\mathcal{M}=0_N$. However, for the interconnection of finite abstractions, which will be constructed later in Subsection \ref{absCur},
		due to possibly different granularities of 
	 internal input and output sets, the designed parameters ${\phi_{ij}}$ are not necessarily zero to make the interconnection well-posed.	 
	Whenever $\phi_{ij} \neq 0$, sets $\mathbb{Y}_{ji}$ are assumed to be finite unions of boxes. 
\end{remark}

For the given system in \eqref{interConnectedsys}, we also denote by  $x \xlongrightarrow{u} x'$ a \textit{transition} in the system if and only if $x' \in f(x,u)$.  For any initial state $x_0 \in X_0$, a finite state run generated from $x_0$ is a finite sequence of transitions
$$x_0 \xlongrightarrow{u_1} x_1 \xlongrightarrow{u_2} \dots \xlongrightarrow{u_{n-1}} x_{n-1} \xlongrightarrow{u_n} x_n,$$
such that $x_i \xlongrightarrow{u_{i+1}} x_{i+1}$ for all $0 \leq i < n$.  A finite output run, also called an \textit{output trajectory}, is a sequence $\{y_0,y_1,\ldots,y_n\}$ such that there exists a finite state run $\{x_0,x_1, \dots, x_n\}$ with $y_i = h(x_i)$, for $i \in\{1,\ldots,n\}$. A finite state and output run can be readily extended to an infinite state and output run as well.	
%
\subsection{Strongly connected components} In this paper, we exploit the interconnection topology of the system and employ knowledge from graph theory as an essential tool in our main results. Here, let us first introduce the terminologes adopted in the paper and recall the notion of strongly connected components, which are used to represent the sub-network of the interconnected system \cite{baier2008principles}.

Consider an interconnected control system $\mathcal{I}_{\mathcal{M}}(\Sigma_1,\dots,\Sigma_N)$ induced by $N \in \N_{\geq 1}$ control subsystems $\Sigma_i$, as defined in Definition \ref{interconnectedsystem}. Hereafter, we denote the directed graph associated with $\mathcal{I}_{\mathcal{M}}(\Sigma_1,\dots,\Sigma_N)$ by $G = (I, E)$, where $I =  [1;N]$ is the set of vertices with each vertex $i \in I$ labelled with subsystem $\Sigma_i$, and $E \subseteq I \times I$ is the set of ordered pairs $(i,j)$, $\forall i,j \in I$, with $y_{ji} \neq 0$.
We denote by $\text{Pre}_I(i) = \{j \in I| \exists (i,j) \in E\}$ as the collection of vertices in $I$ which are direct predecessors of $i$, and by $\text{Post}_I(i) = \{j \in I| \exists (j,i) \in E\}$ as the set of vertices in $I$ which are direct successors of $i$. Intuitively, for any vertex $i$ in graph $G$, the predecessors and successors of $i$ indicate the neighboring subsystems of system $\Sigma_i$. The set $\text{Pre}_I(i)$ is the collection of neighboring subsystems $\Sigma_j$, $j \in I$, which provide internal inputs to subsystem $\Sigma_i$, and set $\text{Post}_I(i)$ is the set of neighboring subsystems $\Sigma_j$, $j \in I$, which accept internal inputs from $\Sigma_i$. Note that because of \eqref{internalinput}, there is no internal input from a subsystem to itself. Therefore, we have $i \notin \text{Pre}_I(i)$ and $i \notin \text{Post}_I(i)$, $\forall i \in I$, which means there is no self-loop for any vertex in $G$. The \textit{strongly connected components} (SCCs) of a directed graph $G$ are maximal strongly connected subgraphs \cite{baier2008principles}. 

In the sequel, we will denote the SCCs in $G$ by $\bar G_k$, $k \in [1;\bar N]$, where $\bar G_k = (I_k, E_k)$ and $\bar N$ is the number of SCCs in $G $. For any $\bar G_k$, we use $\bar N_k$ to denote the cardinality of $I_k$ and set $I_k = \{ k_1,\dots,k_{\bar N_k}\}$. Note that if we regard each SCC as a vertex, the resulting directed graph is acyclic. We denote by $\text{Pre}_G(\bar G_k)$ the collection of SCCs in graph G from which $\bar G_k$ is reachable in one step, by $\text{Post}_G(\bar G_k)$ the collection of SCCs in graph G that is reachable from $\bar G_k$ in one step,
and by $\text{BSCC}(G)= \{\bar G_k, k \in [1;\bar N] |\text{Post}_G(\bar G_k) = \varnothing\}$ the collection of bottom strongly connected components of graph $G$ from which no vertex in $G$ outside $\bar G_k$ is reachable.

\subsection{Approximate opacity for discrete-time control systems}	
Before stating our main results, let us here review notions of approximate opacity proposed in \cite{yin2019approximate}. The adopted notions of secrets are formulated as state-based. In this setting, it is assumed that there exists an intruder (malicious observer) that can only observe the outputs of the systems. Using the observed output information, the intruder aims at inferring the secret states of the system.  Opacity property essentially determines whether or not any trace that reveals secret behaviors of the system is indistinguishable from those, not revealing secret behavior, to an intruder. 
The three basic notions of opacity, i.e. approximate initial-state, current-state, and infinite-step opacity, introduced in \cite{yin2019approximate}, are recalled next.
\begin{definition}
Consider a control system $	\Sigma\!=(\mathbb X,\mathbb X_0,\mathbb X_S,\mathbb U,\mathcal{U},f,\mathbb Y,h)$ and a constant $\delta \geq 0$. System $\Sigma$ is said to be
\begin{itemize}
	\item $\delta$-approximate initial-state opaque if for any $x_0 \in \mathbb X_0 \cap \mathbb X_S$ and finite state run $x_0 \xlongrightarrow{u_1} x_1 \xlongrightarrow{u_2} \dots \xlongrightarrow{u_n} x_n$, there exists $x_0' \in \mathbb X_0 \setminus \mathbb X_S$ and a finite state run $x_0' \xlongrightarrow{u_1'} x_1' \xlongrightarrow{u_2'} \dots  \xlongrightarrow{u_n'} x_n'$ such that 
	$$\max_{i \in [0;n]} \Vert h(x_i)-h(x_i')\Vert \leq \delta.$$
	\item $\delta$-approximate current-state opaque if for any $x_0 \in \mathbb X_0$ and finite state run $x_0 \xlongrightarrow{u_1} x_1 \xlongrightarrow{u_2} \dots \xlongrightarrow{u_n} x_n$ such that $x_n \in \mathbb X_S$, there exists $x_0' \in \mathbb X_0$ and a finite state run $x_0' \xlongrightarrow{u_1'} x_1' \xlongrightarrow{u_2'} \dots \xlongrightarrow{u_n'} x_n'$ such that $x_n' \in \mathbb X \setminus \mathbb X_S$ and 
	$$\max_{i \in [0;n]} \Vert h(x_i)-h(x_i')\Vert \leq \delta.$$
	\item $\delta$-approximate infinite-step opaque if for any $x_0 \in \mathbb X_0$ and finite state run $x_0 \xlongrightarrow{u_1} x_1 \xlongrightarrow{u_2} \dots \xlongrightarrow{u_n} x_n$ such that $x_k \in \mathbb X_S$ for some $k \in [0;n]$, there exists $x_0' \in \mathbb X_0$ and a finite state run $x_0' \xlongrightarrow{u_1'} x_1' \xlongrightarrow{u_2'} \dots \xlongrightarrow{u_n'} x_n'$ such that $x_k' \in \mathbb X \setminus \mathbb X_S$ and 
	$$\max_{i \in [0;n]} \Vert h(x_i)-h(x_i')\Vert \leq \delta.$$
\end{itemize}
\end{definition}
\begin{remark}
Intuitively, the notions of approximate opacity provide a quantitative security guarantee that, if the intruder/observer does not have enough measurement precision, which is captured by the parameter $\delta$, then the secret information of the systems can not be revealed.  
Throughout this work we assume $X_0 \nsubseteq X_S$, otherwise opacity property is trivially violated. Note that we are always interested in verifying opacity of the interconnected systems $\Sigma$ as in Definition \ref{interconnectedsystem} rather than subsystems $\Sigma_i$ introduced in Definition \ref{def:sys1}. The subsystems will be used later in the main compositionality results to show opacity of the interconnected systems.
\end{remark}

\section{Opacity-Preserving Simulation Functions}
In this section, we introduce new notions of approximate opacity-preserving simulation functions, inspired by the notions of simulation functions proposed in \cite{girard2009hierarchical, swikir2019compositional}. The notions of simulation functions in  \cite{girard2009hierarchical, swikir2019compositional} are widely used in abstraction-based techniques to quantify the errors between systems and their abstractions, but without taking into account the opacity properties. The opacity-preserving simulation functions we propose here will play a crucial role in the compositionality results in the next section.   
\subsection{Initial-state opacity-preserving simulation function}	

	First, we introduce a new notion of initial-state opacity-preserving simulation functions. 
	\begin{definition}\label{def:SFD1}
		Consider 
		 $\Sigma_i\!=\!(\mathbb X_i,\!\mathbb X_{i0},\!\mathbb X_{iS},\!\mathbb U_i,\!\mathbb W_i,\!\mathcal{U}_i,\!\mathcal{W}_i,\!f,\!\mathbb Y_i,\!h_i)$ and
		$\hat{\Sigma}_i\!=\!(\hat{\mathbb{X}}_i,\!\hat{\mathbb{X}}_{i0},\!\hat{\mathbb{X}}_{iS},\!\hat{\mathbb{U}}_i,\!\hat{\mathbb{W}}_i,\!\hat{\mathcal{U}}_i,\!\hat{\mathcal{W}}_i,\!\hat{f}_i,\!\hat{\mathbb{Y}}_i,\!\hat{h}_i)$
		where $\hat{\mathbb{W}}_i\subseteq{\mathbb W_i}$ and $\hat{\mathbb{Y}}_i\subseteq{\mathbb Y_i}$. For $\varpi_i \in \mathbb R_{\ge 0}$, a function $V_{i} : \mathbb X_i \times \hat{\mathbb X}_i \rightarrow \mathbb R_{\ge 0}$ is called a $\varpi_i$-approximate initial-state opacity-preserving simulation function ($\varpi_i$-InitSOPSF) from $\Sigma_i$ to $\hat{\Sigma}_i$, if there exist a constant $\vartheta_i \in \mathbb R_{\ge 0}$, and a function $\alpha_i \in \mathcal{K_{\infty}}$ such that 
\begin{enumerate}
	\item[1](a) $\forall x_{i0} \in {\mathbb X}_{i0} \cap {\mathbb X}_{iS}$, $\exists \hat x_{i0} \in \hat {\mathbb X}_{i0} \cap \hat {\mathbb X}_{iS}$, s.t. 
	$V_{i}(x_{i0},\hat x_{i0}) \leq \varpi_i $;\\
	(b) $\forall \hat x_0 \in \hat {\mathbb X}_{i0} \setminus \hat {\mathbb X}_{iS}$, $\exists x_{i0} \in {\mathbb X}_{i0} \setminus {\mathbb X}_{iS}$, s.t.
	$V_{i}(x_{i0},\hat x_{i0}) \leq \varpi_i $;
	\item[2]   $\forall x_i \in \mathbb X_i, \forall \hat x_i \in \hat {\mathbb X}_i$, $\alpha_i(\Vert h_i(x_i) - \hat h_i(\hat x_i) \Vert) \leq V_{i}(x_i,\hat x_i)$;
	\item[3] $\forall x_i \in \mathbb X_i, \forall \hat x_i \in \hat {\mathbb X}_i$ s.t. $V_{i}(x_i,\hat x_i) \leq \varpi_i$, $\forall w_i \in \mathbb W_i$, $\forall \hat w_i \in \hat{\mathbb{W}}_i$ s.t. $\Vert w_i- \hat w_i \Vert \!\leq\! \vartheta_i$, the following conditions hold: \\
    (a) $\forall u_i \in \mathbb U_i$, $\forall x_{id} \in f_i(x_i,u_i,w_i)$, $\exists \hat u_i \in \hat{\mathbb{U}}_i$, $\exists\hat x_{id} \in \hat{f}_i(\hat{x}_i,\hat{u}_i,\hat{w}_i)$, s.t. $V_{i}(x_{id},\hat x_{id}) \leq \varpi_i$;\\
    (b) $\forall \hat u_i \in \hat{\mathbb{U}}_i$, $\forall \hat x_{id} \in \hat{f}_i(\hat{x}_i,\hat{u}_i,\hat{w}_i)$, $\exists u_i \in \mathbb U_i$, $\exists x_{id} \in f_i(x_i,u_i,w_i)$,  s.t. $V_{i}(x_{id},\hat x_{id}) \leq \varpi_i$.
\end{enumerate}
	\end{definition}

If there exists a $\varpi_i$-InitSOPSF from $\Sigma_i$ to $\hat{\Sigma}_i$, and $\hat{\Sigma}_i$ is finite (i.e. $\hat{\mathbb{X}}_i, \hat{\mathbb{U}}_i, \hat{\mathbb{W}}_i$ are finite sets), $\hat{\Sigma}_i$ is called an InitSOP finite abstraction (or symbolic model) of the concrete (original) system $\Sigma_i$, which is constructed later in Subsection \ref{absCur}.
Now, we consider systems without internal inputs, which is the case for interconnected systems (cf. Definition \ref{interconnectedsystem}) and rewrite Definition \ref{def:SFD1} as follows.
	\begin{definition}\label{def:SFD2}
		Consider systems $\Sigma\!=(\mathbb X,\mathbb X_0,\mathbb X_S,\mathbb U,\mathcal{U},f,\mathbb Y,h)$ and $\hat{\Sigma}\!=(\hat{\mathbb{X}},\hat{\mathbb{X}}_0,\hat{\mathbb {X}}_S,\hat{\mathbb{U}},\hat{\mathcal{U}},\hat{f},\hat{\mathbb{Y}},\hat{h})$, where $\hat{\mathbb{Y}}\subseteq{\mathbb{Y}}$. For $\varpi \in \mathbb R_{\ge 0}$, a function $\tilde V: \mathbb X \times \hat{\mathbb{X}} \rightarrow \mathbb R_{\ge 0}$ is called a $\varpi$-approximate initial-state opacity-preserving simulation function ($\varpi$-InitSOPSF) from $\Sigma$ to $\hat{\Sigma}$, if there exist a function $\alpha \in \mathcal{K_{\infty}}$ such that 
		\begin{enumerate}
			\item[1] (a) $\forall x_0 \in {\mathbb X}_0 \cap {\mathbb X}_S$, $\exists \hat x_0 \in \hat {\mathbb X}_0 \cap \hat {\mathbb X}_S$, s.t.
			$\tilde V(x_0,\hat x_0) \leq \varpi $;\\
			(b)  $\forall \hat x_0 \in \hat {\mathbb X}_0 \setminus \hat {\mathbb X}_S$, $\exists x_0 \in {\mathbb X}_0 \setminus {\mathbb X}_S$, s.t.
			$\tilde V(x_0,\hat x_0) \leq \varpi $;
			\item[2]  $\forall x \in \mathbb X, \forall \hat x \in \hat {\mathbb X}$, $\alpha(\Vert h(x) - \hat h(\hat x) \Vert) \leq \tilde V(x,\hat x)$;
			\item[3] $\forall x \in \mathbb X, \forall \hat x \in \hat {\mathbb X}$ s.t. $\tilde V(x,\hat x) \leq \varpi $, the following conditions hold: \\
			(a) $\forall u \in \mathbb U$, $\forall x_d \in f(x,u)$, $\exists \hat u\in \hat{\mathbb{U}}$, $\exists\hat x_d \in \hat{f}(\hat{x},\hat{u})$, s.t. $\tilde V(x_d,\hat x_d) \leq \varpi$;\\
			(b) $\forall \hat u\in \hat{\mathbb{U}}$, $\forall \hat x_d \in \hat{f}(\hat{x},\hat{u})$, $\exists u\in \mathbb U$, $\exists x_d \in f(x,u)$,  s.t. $\tilde V(x_d,\hat x_d) \leq \varpi$.
		\end{enumerate}
	\end{definition}
	If there exists a $\varpi$-InitSOPSF from $\Sigma$ to $\hat{\Sigma}$, and $\hat{\Sigma}$ is finite, $\hat{\Sigma}$ is called an InitSOP finite abstraction of the concrete system $\Sigma$.
			
	Before showing the next result, we recall the definition of $\varepsilon$-approximate initial-state opacity-preserving simulation relation in \cite{yin2019approximate}.
	
		\begin{definition}\label{def:InitSOP}
		Consider systems $\Sigma\!=(\mathbb X,\mathbb X_0,\mathbb X_S,\mathbb U,\mathcal{U},f,\mathbb Y,h)$ and  $\hat{\Sigma}\!=(\hat{\mathbb{X}},\hat{\mathbb{X}}_0,\hat{\mathbb {X}}_S,\hat{\mathbb{U}},\hat{\mathcal{U}},\hat{f},\hat{\mathbb{Y}},\hat{h})$, with $\hat{\mathbb{Y}}\subseteq{\mathbb{Y}}$. For $\varepsilon \in \mathbb R_{\ge 0}$, a relation $R \subseteq \mathbb X \times \hat{\mathbb{X}}$ is called an $\varepsilon$-approximate initial-state opacity-preserving simulation relation ($\varepsilon$-InitSOP simulation relation) from $\Sigma$ to $\hat{\Sigma}$ if 
		\begin{enumerate}
			\item[1] (a) $\forall x_0 \in {\mathbb X}_0 \cap {\mathbb X}_S$, $\exists \hat x_0 \in \hat {\mathbb X}_0 \cap \hat {\mathbb X}_S$, s.t.
			$(x_0,\hat x_0) \in R$;\\
			(b)  $\forall \hat x_0 \in \hat {\mathbb X}_0 \setminus \hat {\mathbb X}_S$, $\exists x_0 \in {\mathbb X}_0 \setminus {\mathbb X}_S$, s.t.
			$(x_0,\hat x_0) \in R$;
			\item[2]  $\forall (x, \hat{x}) \in R$, $\Vert h(x)-\hat h(\hat x)\Vert\leq \varepsilon$;
			\item[3]  For any $(x, \hat{x}) \in R$, we have\\
			(a) $\forall u \in \mathbb U$, $\forall x_d \in f(x,u)$, $\exists \hat u \in \hat{\mathbb{U}}$, $\exists\hat x_d \in \hat{f}(\hat{x},\hat{u})$, s.t. $(x_d,\hat x_d) \in R $; \\
			(b) $\forall \hat u\in \hat{\mathbb{U}}$, $\forall \hat x_d \in \hat{f}(\hat{x},\hat{u})$, $\exists u\in \mathbb U$, $\exists x_d \in f(x,u)$, s.t. $(x_d,\hat x_d) \in R $.
		\end{enumerate}
			We say that a system $\Sigma $ is $\varepsilon$-InitSOP simulated by a system $\hat{\Sigma}$  or a system $\hat{\Sigma}$ $\varepsilon$-InitSOP simulates a system  $\Sigma$, denoted by ${\Sigma} \preceq^{\varepsilon}_{\mathcal{I}}  \hat\Sigma$, if there
	exists an $\varepsilon$-InitSOP simulation relation $R$ from $\Sigma $ to $\hat{\Sigma} $ as in Definition \ref{def:InitSOP}.
	\end{definition}
	
	It is worth noting that the $\varepsilon$-approximate initial-state opacity-preserving simulation relation as in Definition \ref{def:InitSOP} characterizes the distance between two systems in terms of the satisfaction of approximate opacity. This relation not only considers the dynamic, but also the secret property of the system while considering abstractions. The usefulness of Definition \ref{def:InitSOP} in terms of preservation of approximate opacity across related systems is shown in \cite[Theorem 5.2]{yin2019approximate} as stated below.
\begin{proposition}\label{thm:InitSOP}
	Consider systems $\Sigma\!=(\mathbb X,\mathbb X_0,\mathbb X_S,\mathbb U,\mathcal{U},f,\mathbb Y,h)$ and  $\hat{\Sigma}\!=(\hat{\mathbb{X}},\hat{\mathbb{X}}_0,\hat{\mathbb {X}}_S,\hat{\mathbb{U}},\hat{\mathcal{U}},\hat{f},\hat{\mathbb{Y}},\hat{h})$ with the same output sets $\hat{\mathbb{Y}}={\mathbb{Y}}$ and
		let $\varepsilon,\delta\in\mathbb R_{\ge 0}$.
		If  ${\Sigma} \preceq^{\varepsilon}_{\mathcal{I}}  \hat\Sigma$ and $\varepsilon\leq \frac{\delta}{2}$,
		then the following implication holds
		\begin{align}
		&\hat{\Sigma}\tup{ is ($\delta-2\varepsilon$)-approximate initial-state opaque} \nonumber
		\Rightarrow \Sigma \tup{ is $\delta$-approximate initial-state opaque}.\nonumber
		\end{align}
\end{proposition}
		The above implication across two related systems basically provides us a sufficient condition for verifying approximate opacity using abstraction-based techniques. In particular, when confronted with a complex system $\Sigma$, one can do the opacity verification over the simpler system $\hat{\Sigma}$ instead of struggling with system $\Sigma$. 
		
		The next result shows that the existence of an $\varpi$-InitSOPSF for systems without internal inputs (as we introduced in Definition \ref{def:SFD2}) implies the existence of an $\varepsilon$-InitSOP simulation relation between them. 

	\begin{proposition}\label{propinit}
		Consider systems $\Sigma\!=(\mathbb X,\mathbb X_0,\mathbb X_S,\mathbb U,\mathcal{U},f,\mathbb Y,h)$ and  $\hat{\Sigma}\!=(\hat{\mathbb{X}},\hat{\mathbb{X}}_0,\hat{\mathbb {X}}_S,\hat{\mathbb{U}},\hat{\mathcal{U}},\hat{f},\hat{\mathbb{Y}},\hat{h})$, where $\hat{\mathbb{Y}}\subseteq{\mathbb{Y}}$. Assume $\tilde V$ is a $\varpi$-InitSOPSF from $\Sigma$ to $\hat{\Sigma}$ with the corresponding function $\alpha\in\mathcal{K}_\infty$ as in Definition \ref{def:SFD2}. Then, relation $R\subseteq\mathbb{X}\times \hat{\mathbb{X}}$ defined by 
		$$R=\left\{(x,\hat{x})\in \mathbb{X}\times \hat{\mathbb{X}}|\tilde V(x,\hat{x})\leq \varpi \right\},$$ 
		is an $\varepsilon$-InitSOP simulation relation, defined as in Definition \ref{def:InitSOP}, from $\Sigma$ to $\hat{\Sigma}$ with 
		\begin{align}\label{er}
		\varepsilon=\alpha^{-1}(\varpi).
		\end{align}
	\end{proposition}
	\begin{IEEEproof}
		The first condition in Definition \ref{def:InitSOP} follows immediately from condition 1 in Definition \ref{def:SFD2}, i.e. $\tilde V(x_0,\hat x_0) \leq \varpi$. 
		Now, we show that $\forall (x, \hat x) \in R$: $\Vert h(x) - \hat h(\hat x) \Vert \leq \varepsilon$. From condition 2 in Definition \ref{def:SFD2}, one has $\alpha(\Vert h(x) - \hat h(\hat x) \Vert) \leq \tilde V(x,\hat x) \leq \varpi$, which readily results in  
		$\Vert h(x) - \hat h(\hat x) \Vert \leq \alpha^{-1}(\varpi)= \varepsilon$.
		Finally, we show the third condition of $R$. Consider any pair $(x,\hat{x})\in R$, i.e., $\tilde V(x,\hat{x})\leq \varpi$. From 3-a) in Definition \ref{def:SFD2}, one has  $\forall u$, $\forall x_d \in f(x,u)$, $\exists \hat u$, $\exists\hat x_d \in \hat{f}(\hat{x},\hat{u})$ such that $\tilde  V(x_d,\hat x_d) \leq \varpi$. It immediately follows that $(x_d,\hat x_d) \in R$ which satisfies condition 3-a) in Definition \ref{def:InitSOP}. Condition 3-b) can be proved in a similar way, which concludes the proof.  
	\end{IEEEproof}
Given the results of Proposition \ref{thm:InitSOP} and Proposition \ref{propinit}, one can readily see that the existence of a $\varpi$-InitSOPSF from $\Sigma$ to $\hat{\Sigma}$ as in Definition \ref{def:SFD2} implies that, if  $\hat{\Sigma}$ is ($\delta-2\varepsilon$)-approximate initial-state opaque, then $\Sigma$ is $\delta$-approximate initial-state opaque. 
Note that $\varpi_i$-InitSOPSFs introduced in Definition \ref{def:SFD1} are used later for the construction of $\varpi$-InitSOPSFs for the interconnected systems, and they are not used directly for deducing approximate opacity preserving simulation relation.

\subsection{Current-state opacity-preserving simulation function}

Here, we introduce a notion of current-state opacity-preserving simulation functions.
\begin{definition}\label{def:SFDcur}
Consider $\Sigma_i\!=\!(\mathbb X_i,\!\mathbb X_{i0},\!\mathbb X_{iS},\!\mathbb U_i,\!\mathbb W_i,\!\mathcal{U}_i,\!\mathcal{W}_i,\!f,\!\mathbb Y_i,\!h_i)$ and
$\hat{\Sigma}_i\!=\!(\hat{\mathbb{X}}_i,\!\hat{\mathbb{X}}_{i0},\!\hat{\mathbb{X}}_{iS},\!\hat{\mathbb{U}}_i,\!\hat{\mathbb{W}}_i,\!\hat{\mathcal{U}}_i,\!\hat{\mathcal{W}}_i,\!\hat{f}_i,\!\hat{\mathbb{Y}}_i,\!\hat{h}_i)$
where $\hat{\mathbb{W}}_i\subseteq{\mathbb W_i}$ and $\hat{\mathbb{Y}}_i\subseteq{\mathbb Y_i}$.  For $\varpi_i \in \mathbb R_{\ge 0}$, a function $V_{i} : \mathbb X_i \times \hat{\mathbb{X}}_i \rightarrow \mathbb R_{\ge 0}$ is called a $\varpi_i$-approximate current-state opacity-preserving simulation function ($\varpi_i$-CurSOPSF) from $\Sigma_i$ to $\hat \Sigma_i$, if there exist a constant $\vartheta_i \in \mathbb R_{\ge 0}$, and a function $\alpha_i \in \mathcal{K_{\infty}}$ such that 
	\begin{enumerate}
		\item[1] $\forall x_{i0} \in {\mathbb X}_{i0}$, $\exists \hat x_{i0} \in \hat {\mathbb X}_{i0}$, s.t. $V_{i}(x_{i0},\hat x_{i0}) \leq \varpi_i$;
		\item[2]  $\forall x_i \in \mathbb X_i, \forall \hat x_i \in \hat {\mathbb X}_i$, $\alpha_i(\Vert h_i(x_i) - \hat h_i(\hat x_i) \Vert) \leq V_{i}(x_i,\hat x_i)$;
		\item[3] $\forall x_i \!\in\! \mathbb X_i, \forall \hat x_i \!\in\! \hat {\mathbb X}_i$ s.t. $\!V_{i}(x_i,\!\hat x_i) \!\leq\! \varpi_i$, $\forall w_i \!\in\! \mathbb W_i$, $\forall \hat w_i \!\in\! \hat{\mathbb{W}}_i$ s.t. $\!\Vert w_i\!-\! \hat w_i \Vert \!\leq\! \vartheta_i$, the following conditions hold: \\
		(a) $\forall u_i \!\in\! \mathbb U_i$, $\forall x_{id} \!\in\! f_i(x_i,\!u_i,\!w_i)$, $\exists \hat u_i \!\in\! \hat{\mathbb{U}}_i$,  $\exists\hat x_{id} \!\in\! \hat{f}_i(\hat{x}_i,\!\hat{u}_i,\!\hat{w}_i)$, s.t. $V_{i}(x_{id},\!\hat x_{id}) \!\leq\! \varpi_i$;\\
		(b) $\forall u_i \!\in\! \mathbb U_i$, $\forall x_{id} \!\in \!f_i(x_i,\!u_i,\!w_i)$ s.t. $\!x_{id} \!\in \!\mathbb X_{iS}$, $\exists \hat u_i \!\in\! \hat{\mathbb{U}}_i$, $\exists\hat x_{id} \!\in\! \hat{f}_i(\hat{x}_i,\!\hat{u}_i,\!\hat{w}_i)$ with $\hat x_{id} \!\in\! \hat{\mathbb X}_{iS}$, s.t. $\!V_{i}(x_{id},\!\hat x_{id})\! \leq\! \varpi_i$;\\
		(c) $\forall \hat u_i \!\in\! \hat{\mathbb{U}}_i$, $\forall \hat x_{id} \!\in\! \hat{f}_i(\hat{x}_i,\!\hat{u}_i,\!\hat{w}_i)$, $\exists u_i \!\in\! \mathbb U_i$, $\exists x_{id} \in f_i(x_i,\!u_i,\!w_i)$,  s.t. $V_{i}(x_{id},\hat x_{id}) \leq \varpi_i$;\\
		(d) $\forall \hat u_i \!\in\! \hat{\mathbb{U}}_i$, $\forall \hat x_{id} \!\in\! \hat{f}_i(\hat{x}_i,\!\hat{u}_i,\!\hat{w}_i)$ s.t. $\hat x_{id} \in \hat{\mathbb{X}}_i \!\setminus\! \hat{\mathbb X}_{iS}$, $\exists u_i \!\in\! \mathbb U_i$, $\exists x_{id} \!\in\! f_i(x_i,\!u_i,\!w_i)$ with $x_{id} \!\in\! \mathbb X_i \!\setminus\! \mathbb X_{iS}$, s.t. $\!V_{i}(x_{id},\!\hat x_{id})\! \leq\! \varpi_i$.
	\end{enumerate}
\end{definition}

If there exists a $\varpi_i$-CurSOPSF from $\Sigma_i$ to $\hat{\Sigma}_i$, and $\hat{\Sigma}_i$ is finite, $\hat{\Sigma}_i$ is called a CurSOP finite abstraction of the concrete system $\Sigma_i$, which is constructed later in Subsection \ref{absCur}.
For interconnected systems without internal inputs, Definition \ref{def:SFDcur} boils down to the following one.
%
\begin{definition}\label{def:SRDcur2}
	Consider systems $\Sigma\!=(\mathbb X,\mathbb X_0,\mathbb X_S,\mathbb U,\mathcal{U},f,\mathbb Y,h)$ and $\hat{\Sigma}\!=(\hat{\mathbb{X}},\hat{\mathbb{X}}_0,\hat{\mathbb {X}}_S,\hat{\mathbb{U}},\hat{\mathcal{U}},\hat{f},\hat{\mathbb{Y}},\hat{h})$, where $\hat{\mathbb{Y}}\subseteq{\mathbb{Y}}$. For $\varpi \in \mathbb R_{\ge 0}$, a function $\tilde V : \mathbb X \times \hat{\mathbb{X}} \rightarrow \mathbb R_{\ge 0}$ is called a $\varpi$-approximate current-state opacity-preserving simulation function ($\varpi$-CurSOPSF) from $\Sigma$ to $\hat \Sigma$, if there exist a function $\alpha \in \mathcal{K_{\infty}}$ such that 
	\begin{enumerate}
		\item[1]  $\forall x_0 \in {\mathbb X}_0$, $\exists \hat x_0 \in \hat {\mathbb X}_0$, s.t. $\tilde V(x_0,\hat x_0) \leq \varpi $;
		\item[2]  $\forall x \in \mathbb X, \forall \hat x \in \hat {\mathbb X}$, $\alpha(\Vert h(x) - \hat h(\hat x) \Vert) \leq \tilde V(x,\hat x)$;
		\item[3]  $\forall x \in \mathbb X$, $\forall \hat x \in \hat {\mathbb X}$ s.t. $\tilde V(x,\hat x) \leq \varpi$, the following conditions hold: \\
		(a) $\forall u\!\in\! \mathbb{U}$, $\forall x_d \!\in\! f(x,u)$, $\exists \hat u\!\in\! \hat{\mathbb{U}}$, $\exists \hat x_d \!\in\! \hat{f}(\hat{x},\hat{u})$, s.t. $\tilde V(x_d,\hat x_d) \!\leq\! 	\varpi$;\\
		(b) $\forall u\!\in\! \mathbb{U}$, $\forall x_d \!\in\! f(x,u)$ s.t. $x_d \!\in\! \mathbb X_S$, $\exists \hat u\!\in\! \hat{\mathbb{U}}$, $\exists\hat x_d \!\in\! \hat{f}(\hat{x},\hat{u})$ with $\hat x_d \!\in\! \hat{\mathbb X}_S$, s.t. $\tilde V(x_d,\hat x_d) \!\leq\! \varpi$;\\
		(c) $\forall \hat u\!\in\! \hat{\mathbb{U}}$, $\forall \hat x_d \!\in\! \hat{f}(\hat{x},\hat{u})$, $\exists u\!\in\! \mathbb{U}$, $\exists x_d \!\in\! f(x,u)$, s.t. $\tilde V(x_d,\hat x_d) \!\leq\! \varpi$;\\
		(d) $\forall \hat u\!\in\! \hat{\mathbb{U}}$, $\forall \hat x_d \!\in\! \hat{f}(\hat{x},\hat{u})$ s.t. $\hat x_d \!\in\! \hat{\mathbb{X}} \!\setminus\! \hat{\mathbb X}_S$, $\exists u\!\in\! \mathbb{U}$, $\exists x_d \!\in\! f(x,u)$ with $x_d \!\in\! \mathbb X \!\setminus\! \mathbb X_S$, s.t. $\tilde V(x_d,\hat x_d) \!\leq\! \varpi$.
	\end{enumerate}
\end{definition}

If there exists a $\varpi$-CurSOPSF from $\Sigma$ to $\hat{\Sigma}$, and $\hat{\Sigma}$ is finite, $\hat{\Sigma}$ is called a CurSOP finite abstraction of the concrete system $\Sigma$.

Before showing the next result, we recall the definition of $\varepsilon$-approximate current-state opacity-preserving simulation relation defined in \cite{yin2019approximate}.

\begin{definition}\label{def:CURSOP}
	Consider systems $\Sigma\!=(\mathbb X,\mathbb X_0,\mathbb X_S,\mathbb U,\mathcal{U},f,\mathbb Y,h)$ and  $\hat{\Sigma}\!=(\hat{\mathbb{X}},\hat{\mathbb{X}}_0,\hat{\mathbb {X}}_S,\hat{\mathbb{U}},\hat{\mathcal{U}},\hat{f},\hat{\mathbb{Y}},\hat{h})$, where $\hat{\mathbb{Y}}\subseteq{\mathbb{Y}}$. For $\varepsilon \in \mathbb R_{\ge 0}$, a relation $R \subseteq \mathbb X \times \hat{\mathbb{X}}$ is called an $\varepsilon$-approximate current-state opacity-preserving simulation relation ($\varepsilon$-CurSOP simulation relation) from $\Sigma$ to $\hat \Sigma$ if 
	\begin{enumerate}
		\item[1] $\forall x_0 \in {\mathbb X}_0$, $\exists \hat x_0 \in \hat {\mathbb X}_0$, s.t. $(x_0,\hat x_0) \in R$;
		\item[2]  $\forall (x, \hat{x}) \in R$, $\Vert h(x) - \hat h(\hat x) \Vert \leq \varepsilon$;
		\item[3]  For any $(x, \hat{x}) \in R$, we have\\
		(a)  $\forall u\!\in\!\mathbb{U}$, $\forall x_d \!\in\! f(x,u)$, $\exists \hat u\!\in\! \hat{\mathbb{U}}$, $\exists\hat x_d \!\in\! \hat{f}(\hat{x},\hat{u})$, s.t. $(x_d,\hat x_d) \!\in\! R $; \\
		(b) $\forall u\!\in\!\mathbb{U}$, $\forall x_d \!\in\! f(x,u)$ s.t. $x_d \in \mathbb X_S$, $\exists \hat u\!\in\! \hat{\mathbb{U}}$, $\exists\hat x_d \!\in\! \hat{f}(\hat{x},\hat{u})$, s.t. $\hat x_d \!\in\! \hat{\mathbb X}_S$ and $(x_d,\hat x_d) \!\in\! R $; \\
		(c) $\forall \hat u\!\in\! \hat{\mathbb{U}}$, $\forall \hat x_d \!\in\! \hat{f}(\hat{x},\hat{u})$, $\exists u\!\in\!\mathbb{U}$, $\exists x_d \!\in\! f(x,u)$, s.t. $(x_d,\hat x_d) \!\in\! R $;\\
		(d) $\forall \hat u\!\in\! \hat{\mathbb{U}}$, $\forall\hat x_d \!\in\! \hat{f}(\hat{x},\hat{u})$ s.t. $\hat x_d \!\in\! \hat{\mathbb{X}} \!\setminus\! \hat{\mathbb X}_S$, $\exists u\!\in\!\mathbb{U}$, $\exists x_d \!\in\! f(x,u)$,  s.t. $x_d \!\in\! \mathbb X \!\setminus\! \mathbb X_S$ and $(x_d,\hat x_d) \!\in\! R $.
	\end{enumerate}
	We say that a system $\Sigma $ is $\varepsilon$-CurSOP simulated by a system $\hat{\Sigma}$  or a system $\hat{\Sigma}$ $\varepsilon$-CurSOP simulates a system  $\Sigma$, denoted by ${\Sigma} \preceq^{\varepsilon}_{\mathcal{C}}  \hat\Sigma$, if there
exists an $\varepsilon$-CurSOP simulation relation $R$ from $\Sigma $ to $\hat{\Sigma} $.

\end{definition}

The next result shows that the existence of a $\varpi$-CurSOPSF for systems without internal inputs (as defined in Definition \ref{def:SRDcur2}) implies the existence of an $\varepsilon$-CurSOP simulation relation between them. 

\begin{proposition}\label{propcur}
	Consider systems $\Sigma\!=(\mathbb X,\mathbb X_0,\mathbb X_S,\mathbb U,\mathcal{U},f,\mathbb Y,h)$ and  $\hat{\Sigma}\!=(\hat{\mathbb{X}},\hat{\mathbb{X}}_0,\hat{\mathbb {X}}_S,\hat{\mathbb{U}},\hat{\mathcal{U}},\hat{f},\hat{\mathbb{Y}},\hat{h})$, where $\hat{\mathbb{Y}}\subseteq{\mathbb{Y}}$. Assume $\tilde V$ is a $\varpi$-CurSOPSF from $\Sigma$ to $\hat \Sigma$ with the corresponding function $\alpha\in\mathcal{K}_\infty$ as in Definition \ref{def:SRDcur2}. Then, relation $R\subseteq\mathbb{X}\times \hat{\mathbb{X}}$ defined by 
	$$R=\left\{(x,\hat{x})\in \mathbb{X}\times \hat{\mathbb{X}}|\tilde V(x,\hat{x})\leq \varpi\right\},$$ 
	is an $\varepsilon$-CurSOP simulation relation from $\Sigma$ to $\hat \Sigma$ with 
	\begin{align}\label{ercur}
	\varepsilon=\alpha^{-1}(\varpi).
	\end{align}
\end{proposition}
\begin{IEEEproof}
The proof follows the same reasoning as that of Proposition \ref{propinit} and is omitted here. 
\end{IEEEproof}

\subsection{Infinite-step opacity-preserving simulation function}

Now, we introduce a notion of infinite-step opacity-preserving simulation functions by combining the conditions of $\varpi$-InitSOPSF and $\varpi$-CurSOPSF.
\begin{definition}\label{def:SFDInf}
Consider $\Sigma_i\!=\!(\mathbb X_i,\!\mathbb X_{i0},\!\mathbb X_{iS},\!\mathbb U_i,\!\mathbb W_i,\!\mathcal{U}_i,\!\mathcal{W}_i,\!f,\!\mathbb Y_i,\!h_i)$ and
$\hat{\Sigma}_i\!=\!(\hat{\mathbb{X}}_i,\!\hat{\mathbb{X}}_{i0},\!\hat{\mathbb{X}}_{iS},\!\hat{\mathbb{U}}_i,\!\hat{\mathbb{W}}_i,\!\hat{\mathcal{U}}_i,\!\hat{\mathcal{W}}_i,\!\hat{f}_i,\!\hat{\mathbb{Y}}_i,\!\hat{h}_i)$
where $\hat{\mathbb{W}}_i\subseteq{\mathbb W_i}$ and $\hat{\mathbb{Y}}_i\subseteq{\mathbb Y_i}$.  For $\varpi_i \in \mathbb R_{\ge 0}$, a function $V_{i} : \mathbb X_i \times \hat{\mathbb{X}}_i \rightarrow \mathbb R_{\ge 0}$ is called a $\varpi_i$-approximate infinite-step opacity-preserving simulation function ($\varpi_i$-InfSOPSF) from $\Sigma_i$ to $\hat \Sigma_i$, if it is both a $\varpi_i$-InitSOPSF and a $\varpi_i$-CurSOPSF from $\Sigma_i$ to $\hat \Sigma_i$.
\end{definition} 
If there exists a $\varpi_i$-InfSOPSF from $\Sigma_i$ to $\hat{\Sigma}_i$, and $\hat{\Sigma}_i$ is finite, $\hat{\Sigma}_i$ is called an InfSOP finite abstraction of the concrete system $\Sigma_i$, which is constructed later in Subsection \ref{absCur}.
For interconnected systems without internal inputs, Definition \ref{def:SFDInf} boils down to the following one.
\begin{definition}\label{def:SRDInf2}
	Consider systems $\Sigma\!=\!(\mathbb X,\mathbb X_0,\mathbb X_S,\mathbb U,\mathcal{U},f,\mathbb Y,h)$ and $\hat{\Sigma}\!=\!(\hat{\mathbb{X}},\hat{\mathbb{X}}_0,\hat{\mathbb {X}}_S,\hat{\mathbb{U}},\hat{\mathcal{U}},\hat{f},\hat{\mathbb{Y}},\hat{h})$, where $\hat{\mathbb{Y}}\subseteq{\mathbb{Y}}$. For $\varpi \in \mathbb R_{\ge 0}$, a function $\tilde V : \mathbb X \times \hat{\mathbb{X}} \rightarrow \mathbb R_{\ge 0}$ is called a $\varpi$-approximate infinite-step opacity-preserving simulation function ($\varpi$-InfSOPSF) from $\Sigma$ to $\hat \Sigma$, if it is both a $\varpi$-InitSOPSF and a $\varpi$-CurSOPSF from $\Sigma$ to $\hat \Sigma$.
\end{definition}

Before showing the next result, we recall the definition of $\varepsilon$-approximate infinite-step opacity-preserving simulation relation defined in \cite{yin2019approximate}.

\begin{definition}\label{def:InfSOP}
	Consider systems $\Sigma\!=(\mathbb X,\mathbb X_0,\mathbb X_S,\mathbb U,\mathcal{U},f,\mathbb Y,h)$ and  $\hat{\Sigma}\!=(\hat{\mathbb{X}},\hat{\mathbb{X}}_0,\hat{\mathbb {X}}_S,\hat{\mathbb{U}},\hat{\mathcal{U}},\hat{f},\hat{\mathbb{Y}},\hat{h})$, where $\hat{\mathbb{Y}}\subseteq{\mathbb{Y}}$. For $\varepsilon \in \mathbb R_{\ge 0}$, a relation $R \subseteq \mathbb X \times \hat{\mathbb{X}}$ is called an $\varepsilon$-approximate infinite-step opacity-preserving simulation relation ($\varepsilon$-InfSOP simulation relation) from $\Sigma$ to $\hat \Sigma$ if 
	\begin{enumerate}
		\item[1] (a) $\forall x_0 \in {\mathbb X}_0$, $\exists \hat x_0 \in \hat {\mathbb X}_0$, s.t. 
			$(x_0,\hat x_0) \in R$;\\
			(b) $\forall x_0 \in {\mathbb X}_0 \cap {\mathbb X}_S$, $\exists \hat x_0 \in \hat {\mathbb X}_0 \cap \hat {\mathbb X}_S$, s.t. 
			$(x_0,\hat x_0) \in R$;\\
			(c)  $\forall \hat x_0 \in \hat {\mathbb X}_0 \setminus \hat {\mathbb X}_S$, $\exists x_0 \in {\mathbb X}_0 \setminus {\mathbb X}_S$, s.t. 
			$(x_0,\hat x_0) \in R$;
		\item[2]  $\forall (x, \hat{x}) \in R$, $\Vert h(x) - \hat h(\hat x) \Vert \leq \varepsilon$;
		\item[3]  For any $(x, \hat{x}) \in R$, we have\\
		(a)  $\forall u\!\in\!\mathbb{U}$ $\forall x_d \!\in\! f(x,u)$, $\exists \hat u\!\in\! \hat{\mathbb{U}}$, $\exists\hat x_d \in \hat{f}(\hat{x},\hat{u})$, s.t. $(x_d,\hat x_d) \in R $; \\
		(b) $\forall u\!\in\!\mathbb{U}$ $\forall x_d \in f(x,u)$ s.t. $x_d \!\in \!\mathbb X_S$, $\exists \hat u\!\in\! \hat{\mathbb{U}}$, $\exists\hat x_d \!\in\! \hat{f}(\hat{x},\hat{u})$, s.t. $\hat x_d \!\in\! \hat{\mathbb X}_S$ and $(x_d,\hat x_d) \!\in\! R $; \\
		(c) $\forall \hat u\!\in\! \hat{\mathbb{U}}$ $\forall\hat x_d \!\in\! \hat{f}(\hat{x},\hat{u})$, $\exists u\!\in\!\mathbb{U}$, $\exists x_d \!\in\! f(x,u)$, s.t. $(x_d,\hat x_d) \!\in\! R $;\\
		(d) $\forall \hat u\!\in\! \hat{\mathbb{U}}$ $\forall\hat x_d \in \hat{f}(\hat{x},\hat{u})$ s.t. $\hat x_d \!\in\! \hat{\mathbb{X}} \!\setminus\! \hat{\mathbb X}_S$ , $\exists u\!\in\!\mathbb{U}$, $\exists x_d \!\in\! f(x,u)$,  s.t. $x_d \!\in\! \mathbb X \!\setminus\! \mathbb X_S$ and $(x_d,\hat x_d)\! \in\! R $.
	\end{enumerate}
	We say that a system $\Sigma $ is $\varepsilon$-InfSOP simulated by a system $\hat{\Sigma}$  or a system $\hat{\Sigma}$ $\varepsilon$-InfSOP simulates a system  $\Sigma$, denoted by ${\Sigma} \preceq^{\varepsilon}_{\mathcal{IF}}  \hat\Sigma$, if there
exists an $\varepsilon$-InfSOP simulation relation $R$ from $\Sigma $ to $\hat{\Sigma}$.

\end{definition}

%
The next result shows that the existence of an infinite-step opacity-preserving simulation function for systems without internal inputs (as in Definition \ref{def:SRDInf2}) implies the existence of an approximate infinite-step opacity-preserving simulation relation between them. 

\begin{proposition}
	Consider systems $\Sigma\!=(\mathbb X,\mathbb X_0,\mathbb X_S,\mathbb U,\mathcal{U},f,\mathbb Y,h)$ and  $\hat{\Sigma}\!=(\hat{\mathbb{X}},\hat{\mathbb{X}}_0,\hat{\mathbb {X}}_S,\hat{\mathbb{U}},\hat{\mathcal{U}},\hat{f},\hat{\mathbb{Y}},\hat{h})$, where $\hat{\mathbb{Y}}\subseteq{\mathbb{Y}}$. Assume $\tilde V$ is a $\varpi$-InfSOPSF from $\Sigma$ to $\hat \Sigma$ with the corresponding $\mathcal{K}_\infty$ function $\alpha$. Then, relation $R\subseteq\mathbb{X}\times \hat{\mathbb{X}}$ defined as 
	$$R=\left\{(x,\hat{x})\in \mathbb{X}\times \hat{\mathbb{X}}|\tilde V(x,\hat{x})\leq \varpi\right\},$$ 
	is an $\varepsilon$-InfSOP simulation relation from $\Sigma$ to $\hat \Sigma$ with 
	\begin{align}\label{erinf}
	\varepsilon=\alpha^{-1}(\varpi).
	\end{align}
\end{proposition}
\begin{IEEEproof}
	The definition of $\varepsilon$-InfSOP simulation relation is a combination of those of $\varepsilon$-InitSOP simulation relation and $\varepsilon$-CurSOP simulation relation, and likewise for the definition of $\varpi$-InfSOPSF. Thus, the proof of $R$ being $\varepsilon$-InitSOP simulation relation as in Proposition \ref{propinit} and being $\varepsilon$-CurSOP simulation relation as in Proposition \ref{propcur} conclude the proof of this proposition.
\end{IEEEproof}

\section{Compositionality Result}
	In this section, we analyze networks of discrete-time control subsystems. We show how to construct opacity-preserving simulation functions from a network of abstractions to the concrete network compositionally by using local opacity-preserving simulation functions of the subsystems. In fact, utilizing opacity-preserving simulation functions helps us to show the main compositionality result in this section.


\subsection{Compositional construction of abstractions}
In this subsection, we assume that we are given $N$ concrete control subsystems $\Sigma_i\!=(\mathbb X_i,\mathbb X_{i0},\mathbb X_{iS},\mathbb U_i,\mathbb W_i,\mathcal{U}_i,\mathcal{W}_i,f,\mathbb Y_i,h_i)$ together with their corresponding abstractions
$\hat{\Sigma}_i\!=(\hat{\mathbb{X}}_i,\hat{\mathbb{X}}_{i0},\hat{\mathbb{X}}_{iS},\hat{\mathbb{U}}_i,\hat{\mathbb{W}}_i,\hat{\mathcal{U}}_i,\hat{\mathcal{W}}_i,\hat{f}_i,\hat{\mathbb{Y}}_i,\hat{h}_i)$ and opacity-preserving simulation functions $V_i$ from
$\Sigma_i$ to $\hat\Sigma_i$.


The next theorem provides a compositional approach on the construction of abstractions of networks of control subsystems and that of the corresponding initial-state  opacity-preserving simulation functions. 
\begin{theorem}\label{thm:3}
Consider an interconnected control system
$\Sigma=\mathcal{I}_{0_N}(\Sigma_1,\ldots,\Sigma_N)$ induced by
	$N\in\N_{\ge1}$
	control subsystems~$\Sigma_i$. Assume that each $\Sigma_i$ and its abstraction $\hat{\Sigma}_i$ admit a $\varpi_i$-InitSOPSF $V_i$.  Let $\varpi = \max\limits_{i} \varpi_i$ and $\mathcal{\hat M} \in \mathbb{R}^{N \times N}$ be a matrix with elements $\{\mathcal{\hat M}\}_{ii} = 0,\{\mathcal{\hat M}\}_{ij} = \phi_{ij}, \forall i,j \in [1;N], i\neq j $, $0\leq\phi_{ij}\leq \boxspan(\mathbb{\hat Y}_{ji})$. If $\forall i \in [1;N]$ and $\forall j \in \tup{Pre}_I(i)$,
		\begin{align} \label{compoquaninit}
			\alpha^{-1}_{j}(\varpi_j) 	+ \phi_{ij} \leq \vartheta_i, 
		\end{align}
	then, function 
		\begin{align}\label{defVinit}
		\tilde{V}&(x,\hat{x})\Let\max\limits_{i}\{ \frac{\varpi}{\varpi_i} V_i(x_{i},\hat{x}_{i}) \},
		\end{align}
	is a $\varpi$-InitSOPSF  from $\Sigma$ to $\hat \Sigma={\mathcal{ I}_{ \mathcal{\hat M}}}(\hat{\Sigma}_1,\ldots,\hat{\Sigma}_N)$. 
\end{theorem}
\begin{IEEEproof}
	First, we show that condition 1(a) in Definition \ref{def:SFD2} holds. Consider any $x_0=\intcc{x_{10};\ldots;x_{N0}}\in\mathbb X_0\cap \mathbb X_{S}$. For any subsystem $\Sigma_{i}$ and the corresponding $\varpi_i$-InitSOPSF  $V_i$, from the definition of $V_i$, we have $\forall x_{i0} \in \mathbb X_{i0} \cap \mathbb X_{iS}$, $\exists \hat x_{i0} \in \hat{\mathbb{X}}_{i0} \cap \hat{\mathbb{X}}_{iS}$: 
	$V_i(x_{i0},\hat x_{i0}) \leq \varpi_i$. Then, from the definition of $\tilde{V}$ as in $\eqref{defVinit}$ we get $\tilde{V}(x_0, \hat x_0) \leq \varpi$, where $\hat x_0=\intcc{\hat x_{10};\ldots;\hat x_{N0}}\in\hat{\mathbb{X}}_0 \cap \hat{\mathbb{X}}_{S}$. Thus, condition 1(a) in Definition \ref{def:SFD2} holds. 
	Condition 1(b) can be proved in the same way thus is omitted here. 
	Now, we show that condition 2 in Definition \ref{def:SFD2} holds for some $\mathcal{K}_{\infty}$ function $\alpha$.	
		Consider any $x=\intcc{x_1;\ldots;x_N}\in\mathbb X$ and $\hat x=\intcc{\hat x_1;\ldots;\hat x_N}\in\hat{\mathbb{X}}$. Then, using condition 2 in Definition \ref{def:SFD1}, one gets
	\begin{align*}\notag
	\Vert h(x)-\hat{h}(\hat{x})\Vert &=\max\limits_{i}\{\Vert h_{ii}(x_i)-\hat{h}_{ii}(\hat{x}_i)\Vert\}\\&\leq\max\limits_{i}\{\Vert h_{i}(x_i)-\hat{h}_{i}(\hat{x}_i)\Vert\}\leq\max\limits_{i}\{\alpha^{-1}_i(V_i(x_{i},\hat{x}_{i}))\}\leq \hat{\alpha}(\max\limits_{i}\{ \frac{\varpi}{\varpi_i} V_i(x_{i},\hat{x}_{i}) \}),
	\end{align*}
	where $\hat{\alpha}(s)=\max\limits_{i}\{\alpha^{-1}_i(s)\}$ for all $s \in \mathbb{R}_{\geq0}$. By defining $\alpha= \hat{\alpha}^{-1}$, one obtains
	\begin{align*}\notag
	\alpha(\Vert h(x)-\hat{h}(\hat{x})\Vert )\leq \tilde{V}(x,\hat{x}),
	\end{align*}
	which satisfies condition 2 in Definition \ref{def:SFD2}.
	Now, we show that condition 3 holds. Let us consider any 
	$x=\intcc{x_1;\ldots;x_N}\in\mathbb X$ and $\hat x=\intcc{\hat x_1;\ldots;\hat x_N}\in\hat{\mathbb{X}}$ such that $\tilde{V}(x,\hat x) \leq \varpi$. It can be seen that from the construction of $\tilde{V}$ in \eqref{defVinit}, we get $V_i(x_{i},\hat{x}_{i}) \leq \varpi_i$ holds, $\forall i \in[1;N]$. For each pair of subsystems $\Sigma_i$ and $\hat{\Sigma}_i$, the internal inputs satisfy the chain of inequality
		\begin{align*}
		\Vert w_i- \hat{w}_i\Vert =& \max\limits_{j \in \tup{Pre}_I(i)}\{\Vert w_{ij}- \hat{w}_{ij}\Vert \} = \max\limits_{j \in \tup{Pre}_I(i)}\{\Vert y_{ji}-\hat{y}_{ji}+\hat{y}_{ji}-\hat{w}_{ij}\Vert\} 	\leq \max\limits_{j \in \tup{Pre}_I(i)}\{\Vert y_{ji}-\hat{y}_{ji}\Vert + \phi_{ij}\}  \\
	\leq &\max\limits_{j \in \tup{Pre}_I(i)}\{\Vert h_{j}(x_j)\!\!-\!\!\hat{h}_{j}(\hat{x}_j)\Vert  + \phi_{ij}\} 
		\leq\max\limits_{j \in \tup{Pre}_I(i)}\{\alpha^{-1}_{j}( V_j(x_{j},\hat{x}_{j}))  + \phi_{ij}\}
			\leq\max\limits_{j \in \tup{Pre}_I(i)}\{\alpha^{-1}_{j}( \varpi_j)  + \phi_{ij}\}.		
		\end{align*}
	Using \eqref{compoquaninit}, one has $\Vert w_i- \hat{w}_i\Vert \leq \vartheta_i$.
	Therefore, by Definition \ref{def:SFD1} for each pair of subsystems $\Sigma_i$ and $\hat\Sigma_i$, one has $\forall u_i \in \mathbb U_i$ $\forall x_{id} \in f_i(x_i,u_i,w_i)$, there exists $\hat{u}_i\in \hat{\mathbb U}_i$ and $\hat{x}_{id} \in \hat{f}_i(\hat x_i,\hat u_i,\hat w_i)$ such that  $V_i(x_{id},\hat{x}_{id}) \leq \varpi_i$. 
	As a result, we get $\forall u=\intcc{u_{1};\ldots;u_{N}}\in\mathbb U$ $\forall x_d \in f(x,u)$, there exists $\hat u=\intcc{\hat u_{1};\ldots;\hat u_{N}}\in\hat{\mathbb{U}}$ and $\hat{x}_{d} \in \hat{f}(\hat{x},\hat{u})$ such that $\tilde{V}(x_d, \hat{x}_{d}) \Let\max\limits_{i}\{ \frac{\varpi}{\varpi_i} V_i(x_{i},\hat{x}_{i}) \} \leq \varpi$. 
	Therefore, condition 3(a) in Definition \ref{def:SFD2} is satisfied with $\varpi = \max\limits_{i} \varpi_i$. The proof of condition 3(b) uses the same reasoning as that of 3(a) and is omitted here. Therefore, we conclude that $\tilde{V}$ is a $\varpi$-approximate initial-state opacity-preserving simulation function from $\Sigma$ to $\hat \Sigma$. 
\end{IEEEproof}

\begin{remark}
Let us define $\phi_i = [\phi_{i1};\!\ldots\!;\phi_{iN}]$. Note that vectors $\phi_i$ serves later as the internal input quantization parameters
for the construction of finite abstractions of $\Sigma_i$ (see Subsection \ref{absCur}). Moreover, vector $\phi_{i}$ will be designed later in Theorem \ref{Main}.
\end{remark}

Next, we extend the results in Theorem \ref{thm:3} to the case of current-state opacity.

\begin{theorem}\label{thm:compocurSF}
	Consider an interconnected control system
$\Sigma=\mathcal{I}_{0_N}(\Sigma_1,\ldots,\Sigma_N)$ induced by
	$N\in\N_{\ge1}$
	control subsystems~$\Sigma_i$. Suppose that each $\Sigma_i$ admits an abstraction $\hat{\Sigma}_i$ together with a $\varpi_i$-CurSOPSF $V_i$, each associated with constants $\varpi_i,\vartheta_i \in \mathbb R_{\ge 0}$ and function $\alpha_i \in \mathcal{K_{\infty}}$ as in Definition \ref{def:SFDcur}. Let $\varpi = \max\limits_{i} \varpi_i$ and $\mathcal{\hat M} \in \mathbb{R}^{N \times N}$ be a matrix with elements $\{\mathcal{\hat M}\}_{ii} = 0,\{\mathcal{\hat M}\}_{ij} = \phi_{ij}, \forall i,j \in [1;N], i\neq j $, $0\leq\phi_{ij}\leq \boxspan(\mathbb{\hat Y}_{ji})$.  If $\forall i \in [1;N]$ and $\forall j \in \tup{Pre}_I(i)$, inequality \eqref{compoquaninit} holds, then the function defined in \eqref{defVinit} is a $\varpi$-CurSOPSF from $\Sigma$ to $\hat \Sigma={\mathcal{ I}_{\mathcal{\hat M}}}(\hat{\Sigma}_1,\ldots,\hat{\Sigma}_N)$. 
\end{theorem}
\begin{IEEEproof}
	First, we show that condition 1 in Definition \ref{def:SRDcur2} holds. Consider any $x_0=\intcc{x_{10};\ldots;x_{N0}}\in\mathbb X_0$. From Definition \ref{def:SFDcur}, for any subsystem $\Sigma_{i}$ and using the corresponding $\varpi_i$-CurSOPSF $V_i$, $\forall i\in[1;N]$, one has $\forall x_{i0} \in \mathbb X_{i0}$, $\exists \hat x_{i0} \in \hat {\mathbb X}_{i0}$, such that $V_i(x_{i0},\hat x_{i0}) \leq \varpi_i$. Then, from the definition of $\tilde{V}$ in $\eqref{defVinit}$, we get $\tilde{V}(x_0, \hat x_0) \leq \varpi$, where $\hat x_0=\intcc{\hat x_{10};\ldots;\hat x_{N0}}\in\hat{\mathbb{X}}_0$. Thus, condition 1 in Definition \ref{def:SRDcur2} holds. 
	The proof for conditions 2 and 3 in Definition \ref{def:SRDcur2} is similar to that of Theorem \ref{thm:3} and is omitted here. 	
\end{IEEEproof}
Next, we extend the results in Theorem \ref{thm:3} to the case of infinite-step opacity. 

\begin{theorem}\label{thm:compoinfSF}
	Consider an interconnected control system
$\Sigma=\mathcal{I}_{0_N}(\Sigma_1,\ldots,\Sigma_N)$ induced by
	$N\in\N_{\ge1}$
	control subsystems~$\Sigma_i$. Suppose that each $\Sigma_i$ admits an abstraction $\hat{\Sigma}_i$ together with a $\varpi_i$-InfSOPSF $V_i$, each associated with constants $\varpi_i,\vartheta_i \in \mathbb R_{\ge 0}$ and function $\alpha_i \in \mathcal{K_{\infty}}$ as in Definition \ref{def:SFDcur}. Let $\varpi = \max\limits_{i} \varpi_i$ and $\mathcal{\hat M} \in \mathbb{R}^{N \times N}$ be a matrix with elements $\{\mathcal{\hat M}\}_{ii} = 0,\{\mathcal{\hat M}\}_{ij} = \phi_{ij}, \forall i,j \in [1;N], i\neq j $, $0\leq\phi_{ij}\leq \boxspan(\mathbb{\hat Y}_{ji})$.   If $\forall i \in [1;N]$ and $\forall j \in \tup{Pre}_I(i)$, inequality \eqref{compoquaninit} holds, then the function defined in \eqref{defVinit} is a $\varpi$-InfSOPSF from $\Sigma$ to $\hat \Sigma={\mathcal{ I}_{\mathcal{\hat M}}}(\hat{\Sigma}_1,\ldots,\hat{\Sigma}_N)$.
\end{theorem}
\begin{IEEEproof}
	The proof is similar to those of Theorem \ref{thm:3} and Theorem \ref{thm:compocurSF} and is omitted here. 
\end{IEEEproof}

\section{Construction of Symbolic Models}\label{1:IV}
\label{initSym}
In this section, we consider each subsystem $\Sigma\!=(\mathbb X,\mathbb X_0,\mathbb X_S,\mathbb U,\mathbb W,\mathcal{U},\mathcal{W},f,\mathbb Y,h)$ as an infinite, deterministic control system with $\mathbb X_{0} = \mathbb X$. Note that throughout this section, we are mainly talking about subsystems rather than the overall network. However, for the sake of better readability, we often omit index $i$ of subsystems throughout the text in this section. 
We assume the output map $h$ of $\Sigma$ satisfies the following general Lipschitz assumption
\begin{align}\label{lipschitz}
\Vert h(x)-h(x')\Vert\leq \ell(\Vert x-x'\Vert),
\end{align}
for all $x,x'\in \mathbb X$, where $\ell\in\KK$. In addition, the existence of an opacity-preserving simulation function between $\Sigma$ and its finite abstraction is established under the assumption that $\Sigma$ is so-called incrementally input-to-state stable \cite{angeli,tran2019advances} as defined next.
\begin{definition}\label{ass:1}
	System $\Sigma\!=(\mathbb X,\mathbb X_0,\mathbb X_S,\mathbb U,\mathbb W,\mathcal{U},\mathcal{W},f,\mathbb Y,h)$ is called incrementally input-to-state stable ($\delta$-ISS) if there exists a function $ \mathcal{G}:\mathbb X \times \mathbb X \to \mathbb{R}_{\geq0} $  such that $\forall x,x'\in \mathbb X$, $\forall u,u'\in \mathbb U$, $\forall w,w' \in \mathbb W$, the inequalities
	\begin{IEEEeqnarray}{c}	\label{eq:ISTFC1}
		\underline{\alpha} (\Vert x-x'\Vert ) \leq \mathcal{G}(x,x')\leq \overline{\alpha}(\Vert x-x'\Vert ),\\ \label{eq:ISTFC2}
		\mathcal{G}(f(x,u,w),f(x',u',w'))\!-\!\mathcal{G}(x,x') \leq\!-\kappa(\mathcal{G}(x,x'))\!+\!\rho_{int}(\Vert w\!-\! w'\Vert)\!+\!\rho_{ext}(\Vert u\!-\! u'\Vert ),
	\end{IEEEeqnarray}
	hold for some $\underline{\alpha}, \overline{\alpha}, \kappa,\rho_{int},\rho_{ext} \in \mathcal{K}_{\infty}$.
\end{definition}
We additionally assume that there exists a function $\hat{\gamma}\in\mathcal{K}_{\infty}$ such that for any $x,x',x'' \in \mathbb{X}$,
\begin{align}\label{eq:TI}
\mathcal{G}(x,x')\leq \mathcal{G}(x,x'')+\hat{\gamma}(\Vert x'-x''\Vert),
\end{align}
for $\mathcal G$ defined in Definition \ref{ass:1}. Note that in most real applications, the state set $\mathbb{X}$ is a compact subset of $\mathbb{R}^n$ and, hence, condition \eqref{eq:TI} is not restrictive. Interested readers are referred to \cite{zamani2014symbolic} showing how to compute such a function $\hat\gamma$.

\subsection{Construction of finite abstractions}\label{absCur}
Now, we construct a finite abstraction of a $\delta$-ISS control system $\Sigma\!=(\mathbb X,\mathbb X,\mathbb X_S,\mathbb U,\mathbb W,\mathcal{U},\mathcal{W},f,\mathbb Y,h)$.
For the remaining of the paper, we assume that sets $\mathbb X$, $\mathbb X_S$, $\mathbb W$, and $\mathbb U$ are of the form of finite unions of boxes. Consider a concrete control system $\Sigma$ and a tuple $q = (\eta, \theta, \mu, \phi)$ of parameters, where $0 < \eta \leq \tup{min} \{span(\mathbb X_S),span(\mathbb{X} \setminus \mathbb X_S)\}$ is the state set quantization, $0<  \mu  < span (\mathbb{U})$ is the external input set quantization, $\phi$ is a vector containing the internal input set quantization parameters, where $0< \Vert \phi \Vert \leq span (\mathbb{W})$, and $\theta \in \mathbb R_{\ge 0}$ is a design parameter. 
Now a finite abstraction can be defined as
$$\hat{\Sigma}=(\hat{\mathbb{X}},\hat{\mathbb X}_0, \hat{\mathbb X}_S, \hat{\mathbb{U}},\hat{\mathbb{W}},\hat{\mathcal{U}},\hat{\mathcal{W}},\hat{f},\hat{\mathbb{Y}},\hat{h}),$$
where $\hat{\mathbb{X}} = \hat{\mathbb X}_0 = [\mathbb X]_{\eta}$, $\hat{\mathbb X}_S = [\mathbb X_S^{\theta}]_{\eta}$, $\hat{\mathbb{U}} = [\mathbb U]_{\mu}$, $\hat{\mathbb{W}} = [\mathbb W]_{\phi}$, $\hat{\mathbb{Y}} = \{h(\hat x)|\hat x \in \hat{\mathbb{X}}\}$, $\hat{h}(\hat x) = h(\hat x)$, $\forall \hat x \in \hat{\mathbb{X}}$, and 
\begin{align}\label{tranrela}
\hat x_d \in \hat f(\hat x,\hat u, \hat w) \quad \tup{if and only if} \quad \Vert \hat x_d- f(\hat x,\hat u, \hat w) \Vert \leq \eta .
\end{align}

Next, we establish the relation between $\Sigma$ and $\hat{\Sigma}$ via the introduced notions of opacity-preserving simulation functions. 

\subsection{Construction of opacity-preserving simulation functions}
In this subsection, we show that if a finite abstraction $\hat{\Sigma}$ of a $\delta$-ISS $\Sigma$ is constructed with the tuple $q = (\eta, \theta, \mu, \phi)$ of parameters satisfying some conditions, then function $\mathcal{G}$ in Definition \ref{ass:1} is an initial-state (resp. current-state, infinite-step) opacity-preserving simulation function from $\Sigma$ to $\hat \Sigma$.

\begin{theorem}\label{thm:2}
	Let $\Sigma\!=(\mathbb X,\mathbb X,\mathbb X_S,\mathbb U,\mathbb W,\mathcal{U},\mathcal{W},f,\mathbb Y,h)$ be a $\delta$-ISS control system as in Definition \ref{ass:1} with function $\mathcal{G}$ satisfying \eqref{eq:ISTFC1}-\eqref{eq:TI} with corresponding functions $\underline{\alpha}, \overline{\alpha}, \kappa,\rho_{int},\rho_{ext},\hat{\gamma}$. Consider  parameters $\varpi,\vartheta \in \mathbb R_{\ge 0}$. Let $\hat{\Sigma}$ be a finite abstraction as constructed in Subsection \ref{absCur}, 
	with a tuple $q = (\eta, 0, \mu, \phi)$ satisfying 
	\begin{align} \label{secquantinit}	
	\eta \leq \min\{\hat{\gamma}^{-1}[\kappa (\varpi)-\rho_{int}(\vartheta)-\rho_{ext}(\mu)],  \overline{\alpha}^{-1}(\varpi)\}.
	\end{align}
	Then, $\mathcal{G}$ is a $\varpi$-InitSOPSF from $\Sigma$ to $\hat \Sigma$.
\end{theorem}

\begin{proof} We start by proving condition 1 in Definition \ref{def:SFD1}. Consider any initial and secret state $x_0 \in {\mathbb X}_0 \cap {\mathbb X}_S$ in $\Sigma$. Since $\eta\leq span(\mathbb X_S)$, $\mathbb X_S\subseteq\bigcup_{p\in[\mathbb X_S]_{\eta}}\mathcal{B}_{\eta}(p)$, then for every $x\in {\mathbb X}_S$ there always exists $\hat x \in \hat {\mathbb X}_S$ such that $\Vert x-\hat x\Vert\leq\eta$. 
	Hence, there exists $\hat x_0 \in \hat{\mathbb X}_0 \cap \hat{\mathbb X}_S$ with $\mathcal{G}(x_0,\hat x_0) \leq \overline{\alpha} (\Vert x_0-\hat x_0\Vert ) \leq \overline{\alpha}(\eta)$ by \eqref{eq:ISTFC1}, and condition 1(a) in Definition \ref{def:SFD1} is satisfied with $\varpi \geq \overline{\alpha}(\eta)$ by \eqref{secquantinit}. For every $\hat x_0 \in \hat {\mathbb X}_0 \setminus \hat {\mathbb X}_S$, by choosing $x_0 = \hat x_0$ which is also inside ${\mathbb X}_0 \setminus {\mathbb X}_S$, we get $\mathcal{G}(x_0,\hat x_0) = 0 \leq \varpi$. Hence, condition 1(b) in Definition \ref{def:SFD1} holds as well.
	Next, we show that condition 2 in Definition \ref{def:SFD1} holds.
	Since $\Sigma$ is incrementally input-to-state stable as in \eqref{eq:ISTFC1},  and given the Lipschitz assumption on $h$, $\forall x\in \mathbb{X}$ and $ \forall \hat{x} \in \mathbb{\hat{X}}
	$, we have 
	\begin{align}\notag
	\Vert h(x)-\hat{h}(\hat{x})\Vert\leq \ell(\Vert x-\hat{x}\Vert)\leq\ell\circ\underline{\alpha}^{-1}(\mathcal{G}(x,\hat{x})).
	\end{align}	
    Let us define $\alpha=(\ell\circ\underline{\alpha}^{-1})^{-1}$. Then one obtains that condition 2 in Definition \ref{def:SFD1} is satisfied with
	\begin{align}\notag
	\alpha (\Vert h(x)-\hat{h}(\hat{x})\Vert )\leq \mathcal{G}(x,\hat{x}).
	\end{align}	
	Now we show condition 3 in Definition \ref{def:SFD1}. From \eqref{eq:TI}, $\forall x\in \mathbb{X}, \forall \hat{x} \in \mathbb{\hat{X}}, \forall u\in\mathbb{{U}}, \forall \hat{u} \in \mathbb{\hat{U}},\forall w \in \mathbb{W},\forall \hat{w} \in \mathbb{\hat{W}}$, we have for any $\hat{x}_{d}\in\hat{f}(\hat{x},\hat{u},\hat{w})$ 
	\begin{align*}
	\mathcal{G}(x_d,\hat x_d)) \leq\mathcal{G}(x_d,f(\hat{x}, \hat{u},\hat{w})) +\hat{\gamma}(\Vert \hat x_d - f(\hat{x},\hat{u},\hat{w})\Vert),
	\end{align*}
	where\footnote{In this section, we assume that $\Sigma$ is deterministic.} $x_d = f(x,u,w)$.
	From the structure of abstraction, the above inequality reduces to
	\begin{align*}
	\mathcal{G}&(x_d,\hat x_d)\leq\mathcal{G}(x_d,f(\hat{x},\hat{u},\hat{w}))+\hat{\gamma}(\eta).
	\end{align*}
	Note that by \eqref{eq:ISTFC2}, we get 
	\begin{align*}
	\mathcal{G}&(x_d,f(\hat{x},\hat{u},\hat{w}))-\mathcal{G}(x,\hat{x})\leq-\kappa(\mathcal{G}(x,\hat{x}))+\rho_{ext}(\Vert u-\hat u\Vert)+\rho_{int}(\Vert w- \hat{w}\Vert ).
	\end{align*}	
	Hence, $\forall x\in \mathbb{X}, \forall \hat{x} \in \mathbb{\hat{X}}, \forall u\in\mathbb{{U}}, \forall \hat{u} \in \mathbb{\hat{U}}$, $\forall w \in \mathbb{W},\forall \hat{w} \in \mathbb{\hat{W}}
	$, one obtains
	\begin{align}\label{inqual}
	\mathcal{G}&(x_d,\hat x_d)-\mathcal{G}(x,\hat{x})\leq-\kappa(\mathcal{G}(x,\hat{x}))+\rho_{ext}(\Vert u-\hat u\Vert )+\rho_{int}(\Vert w- \hat{w}\Vert)+\hat{\gamma}(\eta),
	\end{align}
	for any $\hat{x}_{d}\in\hat{f}(\hat{x},\hat{u},\hat{w})$. Now, we show condition 3(a) in Definition \ref{def:SFD1}. Let us consider any $x \in \mathbb{X}$ and any $\hat x \in \hat {\mathbb{X}}$ satisfying $\mathcal{G}(x,\hat x) \leq \varpi$, and any $w\in\mathbb{{W}}$ and $\hat w$ such that $\Vert \hat{w}-w\Vert \leq \vartheta$. By the structure of $\hat{\mathbb{U}} = [\mathbb U]_{\mu}$, for any $u\in\mathbb{{U}}$, there always exists $\hat{u}$ satisfying $\Vert \hat{u}-u\Vert \leq \mu$. By combining \eqref{inqual} with \eqref{secquantinit}, for any ${x}_{d} = {f}({x},{u},{w})$ and any 
	$\hat{x}_{d}\in\hat{f}(\hat{x},\hat{u},\hat{w})$, the following inequality holds: 
		\begin{align}\label{sos}
		\mathcal{G}(x_d,\hat x_d) \leq (\mathcal{I}_d-\kappa)(\varpi)+\rho_{ext}(\mu)+\rho_{int}(\vartheta)+\hat{\gamma}(\eta) \leq \varpi.
		\end{align}
%
			Hence, condition 3(a) is satisfied. Similarly, for any $\hat{u}$, by choosing $u=\hat{u}$, for any $\hat{x}_{d}\in\hat{f}(\hat{x},\hat{u},\hat{w})$, condition 3(b) in Definition \ref{def:SFDcur} is also satisfied with $\mathcal{G}(x_d,\hat x_d) \leq (\mathcal{I}_d-\kappa)(\varpi)+\rho_{int}(\vartheta)+\hat{\gamma}(\eta) \leq \varpi$, where ${x}_{d} = {f}({x},{u},{w})$. 
	Therefore, we conclude that $\mathcal{G}$ is a $\varpi$-InitSOPSF from $\Sigma$ to $\hat \Sigma$.	
\end{proof}	
Next, we show a similar result as in Theorem \ref{thm:2}, but for current-state opacity.
\begin{theorem}\label{thm:cursimufunc}
	Let $\Sigma$ be a $\delta$-ISS control system as in Definition \ref{ass:1} with function $\mathcal{G}$ satisfying \eqref{eq:ISTFC1}-\eqref{eq:TI} with corresponding functions $\underline{\alpha}, \overline{\alpha}, \kappa,\rho_{int},\rho_{ext},\hat{\gamma}$. Consider parameters $\varpi,\vartheta \in \mathbb R_{\ge 0}$. Let $\hat{\Sigma}$ be a finite abstraction as constructed in Subsection \ref{absCur}, with a tuple $q = (\eta, \theta, \mu, \phi)$ of parameters satisfying 
\begin{align} \label{secquant}	
\eta \leq \min\{\hat{\gamma}^{-1}[\kappa (\varpi)-\rho_{int}(\vartheta)-\rho_{ext}(\mu)],  \overline{\alpha}^{-1}(\varpi)\} ;  \\ \label{quantconst}
\underline{\alpha}^{-1}(\varpi) \leq \theta.
\end{align}
Then, $\mathcal{G}$ is a $\varpi$-CurSOPSF from ${\Sigma}$ to $\hat{\Sigma}$.
\end{theorem}
\begin{proof} We start by proving condition 1 in Definition \ref{def:SFDcur}. Since $\hat{\mathbb{X}} = \hat{\mathbb X}_0 = [\mathbb X]_{\eta} = [\mathbb X_0]_{\eta}$, $\mathbb X_0 \subseteq\bigcup_{p\in\hat{\mathbb X}_0}\mathcal{B}_{\eta}(p)$, then for every initial state $x_0 \in \mathbb X_0$ in $\Sigma$ there always exists $\hat x_{0} \in \hat{\mathbb X}_0$ in $\widehat{\Sigma}$ such that $\Vert \hat x_0-x_0\Vert \leq \eta$. Hence, one gets $\mathcal{G}(x_0,\hat x_0) \leq \overline{\alpha} (\Vert x_0-\hat x_0\Vert) \leq \overline{\alpha}(\eta)$ by \eqref{eq:ISTFC1}, and by using \eqref{secquant} condition 1 in Definition \ref{def:SFDcur} is satisfied with $\varpi \geq \overline{\alpha}(\eta)$. The proof for conditions 2, 3(a), and 3(c) in Definition \ref{def:SFDcur} is similar to that of Theorem \ref{thm:2}, and is omitted here.

For condition 3(b), let us consider any $u\in\mathbb{{U}}$ s.t.  ${x}_{d} = {f}({x},{u},{w}) \in \mathbb X_S$. Again, by choosing any $\hat{u}$ satisfying $\Vert \hat{u}-u\Vert \leq \mu$, we obtain $\mathcal{G}(x_d,\hat x_d) \leq \varpi$. Additionally, by \eqref{eq:ISTFC1} one gets
 \begin{align}
 \Vert x_d-\hat x_d \Vert \leq \underline{\alpha}^{-1} (\mathcal{G}(x_d,\hat x_d)) \leq \underline{\alpha}^{-1}(\varpi).
 \end{align}
As one can see from the structure of the abstraction, where $\hat{\mathbb X}_S = [\mathbb X_S^{\theta}]_{\eta}$ and using $\theta \geq \underline{\alpha}^{-1}(\varpi)$ in \eqref{quantconst}, from ${x}_{d} \in \mathbb X_S$ one concludes that $\hat x_d \in \hat {\mathbb X}_S$, which shows that condition 3(b) holds as well. Condition 3(d) can be proved similarly, which shows that $\mathcal{G}$ is a $\varpi$-approximate current-state opacity-preserving simulation function from ${\Sigma}$ to $\hat{\Sigma}$.	
\end{proof}
Next, we show a similar result as in Theorem \ref{thm:cursimufunc}, but for infinite-step opacity.

\begin{theorem}\label{thm:infsimufunc}
	Let $\Sigma$ be a $\delta$-ISS control system as in Definition \ref{ass:1} with function $\mathcal{G}$ satisfying \eqref{eq:ISTFC1}-\eqref{eq:TI} with corresponding functions $\underline{\alpha}, \overline{\alpha}, \kappa,\rho_{int},\rho_{ext},\hat{\gamma}$. Consider parameters $\varpi,\vartheta \in \mathbb R_{\ge 0}$. Let $\hat{\Sigma}$ be a finite abstraction as constructed in Subsection \ref{absCur}, with a tuple $q = (\eta, \theta, \mu, \phi)$ of parameters satisfying \eqref{secquant} and \eqref{quantconst}.
	Then, $\mathcal{G}$ is a $\varpi$-InfSOPSF from ${\Sigma}$ to $\hat{\Sigma}$.
\end{theorem}
\begin{proof}
First, note that satisfying conditions \eqref{secquant} and \eqref{quantconst} implies that $\mathcal{G}$ is a current-state opacity-preserving simulation function from ${\Sigma}$ to $\hat{\Sigma}$. The proof is left with showing condition 1 in Definition \ref{def:SFD1} for $\varpi$-InitSOPSF.

Note that by the structure of the abstraction, we have $\hat{\mathbb X}_S = [\mathbb X_S^{\theta}]_{\eta}$, where $\theta \geq \underline{\alpha}^{-1}(\varpi)$ from \eqref{quantconst}.
Consider any initial and secret state $x_0 \in {\mathbb X}_0 \cap {\mathbb X}_S$ in $\Sigma$. Since $\eta\leq span(\mathbb X_S)$, $\mathbb X_S\subseteq\bigcup_{p\in[\mathbb X_S]_{\eta}}\mathcal{B}_{\eta}(p) \subseteq \hat{\mathbb X}_S$, then for every $x\in {\mathbb X_S}$ there always exists $\hat x\in \hat {\mathbb X}_{S}$ such that $\Vert x-\hat x\Vert\leq\eta$. 
Hence, one has $\mathcal{G}(x_0,\hat x_0) \leq \overline{\alpha} (\Vert x_0-\hat x_0\Vert )\leq\eta$ by \eqref{eq:ISTFC1}. Thus, by using \eqref{secquant}, one obtains that condition 1(a) in Definition \ref{def:SFD1} is satisfied with $\varpi \geq \overline{\alpha}(\eta)$. For every $\hat x_0 \in \hat {\mathbb X}_0 \setminus \hat {\mathbb X}_S$, by choosing $x_0 = \hat x_0$ which is also inside ${\mathbb X}_0 \setminus {\mathbb X}_S$, we get $\mathcal{G}(x_0,\hat x_0) = 0 \leq \varpi$. Hence, condition 1(b) in Definition \ref{def:SFD1} holds as well, which concludes the proof.
\end{proof}

One can observe that in order to satisfy conditions \eqref{compoquaninit} and \eqref{secquantinit} (resp. \eqref{secquant}) simultaneously, the interconnected system must hold some property. Otherwise, those conditions may not hold at the same time. Before stating the next main result, we consider the following assumption which provides a small-gain type condition such that one can verify whether those competing conditions can be satisfied simultaneously.

\begin{assumption} \label{smallgainSCC}
	Consider an interconnected control system
$\Sigma=\mathcal{I}_{0_N}(\Sigma_1,\ldots,\Sigma_N)$ induced by $N\in\N_{\ge1}$ $\delta$-ISS control subsystems~$\Sigma_i$ which is associated with a directed graph  $G$. Assume that each $\Sigma_i$ and its abstraction $\hat{\Sigma}_i$ admit an initial-state (resp. current-state, infinite-step) opacity-preserving simulation function $\mathcal{G}_i$, together with functions $\kappa_i$, $\alpha_i$, $\bar \alpha_{i}$, and $\rho_{inti}$ as appeared in Definition \ref{def:SFD1} (resp. Definition \ref{def:SFDcur}, Definition \ref{def:SFDInf}) and Definition \ref{ass:1}. Let $\bar G_k = (I_k, E_k)$, $k \in [1;\bar N]$, be the SCCs in $G$, with each $\bar G_k$ consists of $\bar N_k \in \N_{\geq 1}$ vertices, $\sum_{k=1}^{\bar N}\bar N_k = N$, where each vertex represents a control subsystem. For any $\bar G_k$, we define $\forall i,j \in I_k$, 
	\begin{align}\label{gammadcur}
	~\gamma_{ij} = \left \{ \begin{array}{cc} 
	\kappa_{i}^{-1}\circ\rho_{inti}\circ\alpha_{j}^{-1}& \mbox{if } j \in \tup{Pre}_{I_k}(i),\\
	0 & \mbox{otherwise},
	\end{array}\right.
	\end{align}
where $\tup{Pre}_{I_k}(i) = \{j \in I_k| \exists (i,j) \in E\}$.
We assume that for every $\bar G_k$, $k \in [1;\bar N]$, the following small-gain type condition holds
\begin{align}\label{smallgaincur}
\gamma_{i_1i_2}\circ\gamma_{i_2i_3}\circ\cdots\circ\gamma_{i_{r-1}i_r}\circ\gamma_{i_ri_1}<\mathcal{I}_d,
\end{align}
$\forall(i_1,\ldots,i_r)\in\{k_1,\ldots,k_{\bar N_k}\}^r$, where $r\in \{1,\ldots,\bar N_k\}$.

\end{assumption}

\begin{algorithm}[ht]
	\DontPrintSemicolon
	
	\KwInput{The desired precision $\varpi \in \mathbb{R}_{>0}$; the directed graph $G$ composed of SCCs $\bar G_k$, $\forall k \in [1;\bar N]$; the simulation functions $\mathcal{G}_i$ equipped with functions $\kappa_i$, $\alpha_i$, and $\rho_{inti}$, $\forall i \in [1;N]$; functions $\sigma_{k_i}$ $\forall i \in I_k$ satisfying \eqref{gamcur} for $\bar G_k$, $\forall k \in [1;\bar N]$.}
	\KwOutput{$\varpi_i \in \mathbb{R}_{>0}$ and $\vartheta_i \in \mathbb{R}_{>0}$, $\forall i \in [1;N]$}

	Set $\varpi_i := \infty$, $\vartheta_i := \infty$, $\forall i \in [1;N]$, $\tup{sign}_k^* = 0$, $\forall k \in [1;\bar N]$, $G^* =  G$, 
	
	\While{$G^* \neq \varnothing$}
	{
		\ForEach{$\bar G_k \in \tup{BSCC}(G^*)$ }
		{
			\If{$\tup {sign}$$_k^* = 0$}
			{ $\tup{sign}_k^* = 1$
				
				\eIf{$G^* =  G$}
				{
					\tcc{Graph $G$ represents the entire network}
					\eIf{$\bar N_k > 1$} 
					{choose $r \in \mathbb{R}_{>0}$ s.t. $\max\limits_{i \in I_k}\{\sigma_i(r)\}= \varpi$;
						
					set $\varpi_i\!\!=\!\!\sigma_i(r)$, choose $\phi_{ij}$ s.t. $\!\!\max\limits_{j\in\tup{Pre}_{I_k}\!(i)\!}\!\{\phi_{ij}\}\! <\! \rho^{-1}_{inti}\!\!\circ\!\kappa_{i}(\varpi_i)\!-\!\!\!\!\!\max\limits_{j \in \tup{Pre}_{I_k}(i)}\{\!\alpha^{-1}_j(\varpi_j)\!\}, \forall i,j \!\in\! I_k$,					
					set $\vartheta_i\! \!=\!\! \max\limits_{j\in \tup{Pre}_{I_k}\!(i)\!}\!\{\!\alpha^{-1}_j\!(\varpi_j)  \!+\! \phi_{ij}\}$, $\forall i \!\in\! I_k$,				
					and choose $\phi_{ij} < \vartheta_i, \forall i \in I_k, \forall j\in \tup{Pre}_{I\setminus I_k}(i)$;		
					}
					{
						\tcc{The SCC contains only 1 subsystem}
						set $\varpi_i\! =\! \varpi$, choose $\vartheta_i \!\in\! \mathbb{R}_{>0}$ s.t. $\vartheta_i \!<\! \rho_{inti}^{-1}\circ\kappa_{i}(\varpi_i)$, $i \!\in\! I_k$; 
						
					choose $\phi_{ij} < \vartheta_i, \forall i \in I_k, \forall j\in \tup{Pre}_{I\setminus I_k}(i)$;
					}
				}
				{
					\eIf{$\bar N_k > 1$}
					{	
						choose $r \in \mathbb{R}_{>0}$, s.t. $\sigma_i(r) \leq \alpha_{i}(\min\limits_{j \in \tup{Post}_{I \setminus I_k}(i)}\{\vartheta_j\!-\!\phi_{ji}\})$, $\forall i \in I_k$ with $\tup{Post}_{I \setminus I_k}(i) \neq \varnothing$;	
				
						set $\varpi_i\!\!=\!\!\sigma_i(r)$, choose $\phi_{ij}$ s.t. $\!\!\max\limits_{j\in \tup{Pre}_{I_k}\!(i)\!}\!\{\phi_{ij}\}\! <\! \rho^{-1}_{inti}\!\!\circ\!\kappa_{i}(\varpi_i)\!-\!\!\!\!\max\limits_{j \in \tup{Pre}_{I_k}(i)}\{\!\alpha^{-1}_j(\varpi_j)\}, \forall i,j \!\in\! I_k,$ set $\vartheta_i \!\!=\!\! \max\limits_{j \in \tup{Pre}_{I_k}\!(i)\!}\{\!\alpha^{-1}_j(\varpi_j) \!+\! \phi_{ij}\}$, $\forall i \!\in\! I_k$,
						and choose $\phi_{ij} < \vartheta_i, \forall i \in I_k, \forall j\in \tup{Pre}_{I\setminus I_k}(i)$;
					}
					{
						\tcc{The SCC contains only 1 subsystem}
						set $\varpi_i \!\leq\! \alpha_{i}(\min\limits_{j \in \tup{Post}_{I \setminus I_k}(i)}\{\vartheta_j\!-\!\phi_{ji}\})$ and choose $\vartheta_i \!\in\! \mathbb{R}_{>0}$ s.t. $\vartheta_i \!<\! \rho_{inti}^{-1}\circ\kappa_{i}(\varpi_i)$, $i \!\in\! I_k$; 
						
						choose $\phi_{ij} < \vartheta_i, \forall i \in I_k, \forall j\in \tup{Pre}_{I\setminus I_k}(i)$;
					}
				}
			}
		}
		
		$G^* = G^* \setminus \tup{BSCC}(G^*)$;
	}
	\caption{Compositional design of local parameters $\varpi_i \in \mathbb{R}_{>0}$ and $\vartheta_i \in \mathbb{R}_{>0}$, $\forall i \in [1;N]$} \label{quantialgo}
\end{algorithm}

Now, we provide the next main theorem showing that under the above assumption, one can always compositionally design local quantization parameters such that conditions \eqref{compoquaninit} and \eqref{secquantinit} (resp. \eqref{secquant}) are fulfilled simultaneously.

\begin{theorem} \label{Main}
	Suppose that Assumption \ref{smallgainSCC} holds. Then, 
		for any desired precision $\varpi \in \mathbb{R}_{>0}$ as in Definition \ref{def:SFD2} (resp. Definition \ref{def:SRDcur2}, Definition \ref{def:SRDInf2}), 
		there always exist quantization parameters $\eta_i, \mu_i, \phi_i$, $\forall i \!\in\! [1;N]$,  such that \eqref{compoquaninit} and \eqref{secquantinit} (resp. \eqref{secquant}) are satisfied simultaneously, where the local parameters $\vartheta_i \!\in\! \mathbb{R}_{>0}$ and $\varpi_i \!\in\! \mathbb{R}_{>0}$,  $\forall i \!\in\! [1;N]$, are obtained from Algorithm \ref{quantialgo}. 
	\begin{proof}
	First, let us note that the small-gain type condition \eqref{smallgaincur} implies that for each $\bar G_k$, there exists $ \sigma_i \in \mathcal{K}_{\infty}$ satisfying, $\forall i \in I_k$,  
		\begin{align}\label{gamcur}
		&\max\limits_{j \in \tup{Pre}_{I_k}(i)}\{\gamma_{ij}\circ\sigma_j\}<\sigma_i;
 		\end{align} 
		see \cite[Theorem 5.2]{090746483}. Now, given a desired precision $\varpi$, we apply Algorithm \ref{quantialgo} to design the pair of parameters $(\varpi_i, \vartheta_i)$, $\forall i\in [1;N]$, for all of the subsystems.
	    In order to show that the algorithm guarantees the simultaneous satisfaction of conditions \eqref{compoquaninit} and \eqref{secquantinit} (resp. \eqref{secquant}), let us consider different scenarios of the SCCs. First, we consider the SCCs which are composed of only 1 subsystem, i.e $\bar N_k =1$. From lines 11 and 19, one observes that the selections of $\varpi_i$ and $\vartheta_i$ for each subsystem immediately ensure that
		\begin{align}
		\kappa_{i}(\varpi_i)-\rho_{inti}(\vartheta_i) >0,
     	\end{align}
		which implies that there always exist quantization parameters $\eta_i, \mu_i$ to satisfy \eqref{secquantinit} (resp. \eqref{secquant}).
		Next, let us consider the SCCs with more than 1 subsystems, i.e $\bar N_k >1$.		Now, suppose that for each $\bar G_k$, we are given a sequence of functions $\sigma_i \in \mathcal{K}_{\infty}$, $\forall i\in I_k$, satisfying \eqref{gamcur}. 
		From \eqref{gammadcur} and \eqref{gamcur}, we have, $\forall i\in I_k$,
		\begin{align} \label{inequalityinscc}
		\!\!\!\max\limits_{j \in \tup{Pre}_{I_k}\!(i)\!}\{\gamma_{ij}\!\circ\!\sigma_j\}\!<\!\sigma_i
		\!\!\Longrightarrow\! \!\!\!\max\limits_{j \in \tup{Pre}_{I_k}\!(i)\!}\{\kappa_{i}^{-1}\!\circ\!\rho_{inti}\!\circ\!\alpha_{j}^{-1}\!\! \circ\!\sigma_j\}\!<\!\sigma_i \!\Longrightarrow \!\rho_{inti} \!\circ\!\!\!\!\max\limits_{j \in \tup{Pre}_{I_k}\!(i)\!}\{\alpha_{j}^{-1}\!\! \circ\!\sigma_j\}\!< \!\kappa_{i}\!\circ\!\sigma_i.
		\end{align}
		
		Now, let us set $\varpi_i=\sigma_i(r)$, $\forall i \in I_k$, where $r$ is chosen under the criteria in lines 8 and 16,  and choose the internal input quantization parameters $\phi_{ij}$ such that
		\begin{align} \label{wquanti}
		 \max\limits_{j \in \tup{Pre}_{I_k}(i)}\{\phi_{ij}\} < \rho^{-1}_{inti}\circ\kappa_{i}(\varpi_i)-\max\limits_{j \in \tup{Pre}_{I_k}(i)}\{\alpha^{-1}_j(\varpi_j)\}, \forall i,j \in I_k.
		\end{align}
	Now, by setting $\vartheta_i\! =\! \max\limits_{j \in \tup{Pre}_{I_k}\!(i)\!}\{\alpha^{-1}_j(\varpi_j) \!+\! \phi_{ij}\}$ and combining  \eqref{wquanti} with \eqref{inequalityinscc}, one has, $\forall i \in I_k$,
			\begin{align}\label{posiinequal}
		\rho_{inti}(\vartheta_i) = \rho_{inti}(\max\limits_{j \in \tup{Pre}_{I_k}(i)} \{\alpha^{-1}_j(\varpi_j)  + \phi_{ij}\})\leq \rho_{inti} (\max\limits_{j \in \tup{Pre}_{I_k}(i)}\{\alpha^{-1}_j(\varpi_j)\}+\max\limits_{j \in \tup{Pre}_{I_k}(i)}\{\phi_{ij}\})< \kappa_{i}(\varpi_i),
		\end{align}
	which again implies that one can always find suitable local parameters $\eta_i, \mu_i$ to satisfy \eqref{secquantinit} (resp. \eqref{secquant}).
Additionally, it can be observed that, the design procedure in Algorithm \ref{quantialgo} follows the hierarchy of the acyclic directed graph which is composed of SCCs as vertices. 
	The selection of $\vartheta_i =  \max\limits_{j \in \tup{Pre}_{I_k}(i)} \{\alpha^{-1}_j(\varpi_j)  + \phi_{ij}\}$ as in lines 9 and 17, together with the design procedure for $\varpi_i$ and $\phi_{ij}$ ensure that  \eqref{compoquaninit}  is satisfied as well, which concludes the proof.
	\end{proof}
\end{theorem}

\begin{remark}
Note that by involving the notion of SCCs in the design procedure for selecting parameters, we are allowed to check the small-gain condition and design local parameters inside each SCC, independently of the entire network. Let us also remark the soundness of Algorithm \ref{quantialgo}. It can be seen that as long as Assumption \ref{smallgainSCC} holds, for any desired precision $\varpi \in \mathbb{R}_{>0}$, the algorithm always provides us suitable pairs of parameters $(\varpi_i, \vartheta_i)$, $\forall i\in [1;N]$. In addition, as can be observed in the procedure done in lines 9-12 and 16-20 of Algorithm \ref{quantialgo}, provided that the inequalities hold, we have freedom to choose the local parameters according to local quantization criteria. Note that since the interconnected system we consider in this paper is composed of finite number of subsystems, the number of SCCs is finite. Therefore, the algorithm terminates in finite iterations. Note that the SCCs of a graph $G$ can be computed in O(m) time \cite{Tarjan72depthfirst}, where $m$ is the number of edges in $G$.  
\end{remark}

\section{Example}
\subsection{Compositional construction of opacity-preserving finite abstractions}
Here, we provide an illustrative example to explain the design procedure of local quantization parameters using Algorithm \ref{quantialgo}. The system model is adapted from \cite{pola2016symbolic}.

Consider the interconnected discrete-time system $\Sigma$ consisting of $n= 6$ subsystems:
\begin{align}\label{case1sys}
\Sigma:\left\{
\begin{array}[\relax]{rl}
\mathbf{x}_1(k+1)\!\!\!\!\!&=k_{11}\frac{\mathbf{x}_1(k)}{1+\mathbf{x}^2_1(k)}+\nu_1(k),\\
\mathbf{x}_2(k+1)\!\!\!\!\!&=k_{21}\tanh (\mathbf{x}_2(k))+k_{22}(\textup{sech}(\mathbf{x}_3(k))-1+\mathbf{x}_1(k)),\\
\mathbf{x}_3(k+1)\!\!\!\!\!&=k_{31}\mathbf{x}_3(k)+k_{32}(\sin \mathbf{x}_2(k) + \mathbf{x}_5(k)) +\nu_3(k),\\
\mathbf{x}_4(k+1)\!\!\!\!\!&=k_{41}(\cos(\mathbf{x}_4(k))-1)+k_{42}(\tanh(\mathbf{x}_5(k))),\\
\mathbf{x}_5(k+1)\!\!\!\!\!&=k_{51}\sin(\mathbf{x}_5(k))+k_{52}(\textup{sech}(\mathbf{x}_4(k))-1)+\nu_5(k),\\
\mathbf{x}_6(k+1)\!\!\!\!\!&=k_{61}\frac{\mathbf{x}_6(k)}{1+|\mathbf{x}_6(k)|}+k_{62}\mathbf{x}_5(k),\\
\mathbf{y}(k)\!\!\!\!\!&=\mathbf{x}(k),
\end{array}\right.
\end{align}

where $k \in \N$, $\mathbf{x}(k) = [\mathbf{x}_1(k);\dots;\mathbf{x}_n(k)]$, $\mathbf{y}(k)=[\mathbf{x}_1(k);\dots;\mathbf{x}_n(k)]$. The outputs of the subsystems are:  $\mathbf{y}_i(k)\!=\!c_i\mathbf{x}_i(k)$, where $c_i=[c_{i1};\dots;c_{in}]$ with $c_1 = [1;1;0;0;0;0]$, $c_2 = [0;1;1;0;0;0]$, $c_3 = [0;1;1;0;0;0]$, $c_4 = [0;0;0;1;1;0]$, $c_5 = [0;0;1;1;1;1]$, $c_6 = [0;0;0;0;0;1]$, internal inputs subject to the constraints $w_i =[y_{1i};\ldots;y_{(i-1)i};y_{(i+1)i};\ldots;y_{ni}]$, $\forall i \in [1;6]$,
$\kappa_{i,1}=0.4$, $\forall i \in [1;6]$, $\kappa_{i,2} =0.2$, $\forall i \in [2;5]$,  $\mathbb{X}_i = [-1,1]$ and $\mathbb{U}_i = [-1,1]$, $\forall i \in [1;6]$. One can readily verify that the system $\Sigma$ in \eqref{case1sys} can be seen as an interconnection of 6 scalar subsystems $\Sigma_i$, $i \in [1;6]$, as in Definition \ref{interconnectedsystem}.
The directed graph $G = (I,E)$ is specified by $I = [1;6]$, $E = \{(2,1),(3,2),(2,3),(5,4),(3,5),(4,5),(6,5)\}$. Strongly connected components of $G$ are $\bar G_1$ with $I_1 = \{1\}$, $\bar G_2$ with $I_2 = \{4,5\}$, $\bar G_3$ with $I_3 = \{2,3\}$ and $\bar G_4$ with $I_4 = \{6\}$.
Now we apply our main results in the previous sections to compositionally construct a finite abstraction of $\Sigma$ with accuracy $\varepsilon=0.01$ as defined in \eqref{er}, which preserves approximate initial-state opacity.

First, let us choose functions $V_i = |x_i- x'_i|$, $\forall i \in [1;6]$. It can be readily seen that $V_i$ are $\delta$-ISS Lyapunov functions for subsystems $\Sigma_i$ satisfying \eqref{eq:ISTFC1} and \eqref{eq:ISTFC2} in Definition \ref{ass:1}, with $\kappa_i (s) = (1-|\kappa_{i,1}|)s $, $\underline{\alpha}_i (s)= \overline{\alpha}_i (s)={\hat \gamma}_i (s) = s$, ${\rho_{int}}_1(s) = 0$, ${\rho_{int}}_2(s) = 2|\kappa_{2,2}|s$, ${\rho_{int}}_3(s) =2|\kappa_{3,2}|s$, ${\rho_{int}}_4(s) = |\kappa_{4,2}|s$, ${\rho_{int}}_5(s) = |\kappa_{5,2}|s$, ${\rho_{int}}_6(s) = |\kappa_{6,2}|s$,
${\rho_{ext}}_2(s) = {\rho_{ext}}_4(s) = {\rho_{ext}}_6(s) = 0$, ${\rho_{ext}}_1(s) = {\rho_{ext}}_3(s) ={\rho_{ext}}_5(s) =s$.
The Lipschitz assumption holds with $\ell_i (s) = s$. Since we have $\gamma_{ij}(s) < \mathcal{I}_d$ as defined in \eqref{gammadcur}, $\forall i,j \in I$, the small-gain condition \eqref{smallgaincur} is readily satisfied for every SCC. Functions $\sigma_i = \mathcal{I}_d$, $\forall i \in I$, readily satisfy \eqref{gamcur}. 

Now we apply Algorithm \ref{quantialgo} to design the local parameters. The desired precision is $\varpi=0.01$ by \eqref{er}. We design for all of the subsystems, $\phi_{i}=0$, $\forall i \in [1;6]$. We start with $G^* = G$ and get the associated $\tup{BSCC}(G^*) = \{\bar G_3, \bar G_4\}$ for line 3. First, let us consider the SCC $\bar G_3$. We choose $r = 0.01$ to satisfy the conditions in lines $8-9$ with $\varpi_2 = \vartheta_2 = \varpi_3 = \vartheta_3 = \varpi=0.01$. For $\bar G_4$, since it contains only 1 subsystem $\Sigma_6$, we get in line $11$, $\varpi_6 = \varpi =0.01$ and choose $\vartheta_6 = 0.01$. Now $G^*$ is updated in line $25$ to $\{\bar G_1, \bar G_2\}$. The bottom SCCs of the updated $G^*$ is $\{\bar G_1, \bar G_2\}$. Since the current graph $G^*$ is not the entire network anymore, we go to lines $16-20$. We proceed with $\bar G_1$ firstly. Since $\bar G_1$ consists of only 1 subsystem, we go to line $19$ and set $\varpi_1 = \vartheta_2 = 0.01$ and $\vartheta_1 = 0.01$ such that the inequalities hold. Now consider $\bar G_2$. In line $16$, we choose $r = \min\{\vartheta_3, \vartheta_6\} = 0.01$, and then set $\varpi_4 = \varpi_5 = r$, $\vartheta_4 = \varpi_5$ and $\vartheta_5 = \varpi_4$ in line $17$. Next, the set $G^*$ becomes empty and the algorithm ends. Till now, we obtain local parameters $(\varpi_i, \vartheta_i)$ for each subsystem. Now we have the freedom to design the local quantization parameters $\eta_i, \mu_i$ using $(\varpi_i, \vartheta_i)$ while satisfying inequality \eqref{secquantinit}. We show here a choice of suitable tuples of local parameters $q_i = (\eta_i, \theta_i, \mu_i, \phi_i)$ as: $q_1 = (0.006,0,0,0)$, $q_2 = (0.002,0,0,0)$, $q_3 = (0.002,0,0,0)$, $q_4 = (0.004,0,0,0)$, $q_5 = (0.004,0,0,0)$, $q_6 = (0.004,0,0,0)$. Now, one can construct local abstractions for subsystems as in Subsection \ref{absCur}. Using the result in Theorem \ref{thm:2}, one can verify that $V_i = |x_i- x'_i|$ is a $\varpi_i$-InitSOPSF from each $\Sigma_i$ to its abstraction $\hat \Sigma_i$.
By the results in Theorem \ref{thm:3}, one can verify that 	$\tilde{V}(x,\hat x)=  \max\limits_{i}\{|x_i- \hat{x}_i| \}$ is a $\varpi$-InitSOPSF from $\Sigma$ to $\hat \Sigma = \mathcal{I}_{0_n}(\hat \Sigma_1,\dots,\hat \Sigma_n)$.

\subsection{Compositional verification of initial-state opacity for an interconnected system}

Consider the interconnected discrete-time linear system $\Sigma$ described by:
\begin{align}\label{examplelin}
\Sigma:\left\{
\begin{array}[\relax]{rl}
\mathbf{x}(k+1)\!\!\!\!\!&=A\mathbf{x}(k)+B\nu(k),\\
\mathbf{y}(k)\!\!\!\!\!&=C\mathbf{x}(k),
\end{array}\right.
\end{align}
where  $k \in \N$, $A\in\R^{n\times n}$ is a matrix with $\{A\}_{ii} = a_i = 0.1$, $\{A\}_{i(i-1)} = 0.05$, $\forall i \in [2;n]$, and all other elements are zero, $B\in\R^{n\times n}$ is a diagonal matrix with $\{B\}_{ii} = b_i = 1$, $\{B\}_{ij} = 0$, $\forall i, j \in [1;n], i \neq j$, $C=[0~0~\ldots~0~1] \in\R^{1 \times n}$,
$\mathbf{x}(k) = [\mathbf{x}_1(k);\dots;\mathbf{x}_n(k)]$, $\nu(k) = [\nu_1(k);\dots;\nu_n(k)]$, and $\mathbf{y}(k) = \mathbf{y}_{nn}(k)$.
Intuitively, the output of the overall system is the external output of the last subsystem $\Sigma_n$. The state space is $\mathbb X =\mathbb X_0 = ]0~ 0.6[^n$, the input set is a singleton $\mathbb U = \{0.145\}^n$ and the secret set is $\mathbb{X}_S = ]0~0.2] \times [0.4~0.6[ \times ]0~0.6[^{n-2}$, and the output set is $\mathbb Y = ]0~0.6[$.

Now, let us consider $n \in \N_{\geq 1}$ subsystems $\Sigma_i$, each described by:
\begin{align}\label{exsm}
\Sigma_i:\left\{
\begin{array}[\relax]{rl}
\mathbf{x}_i(k+1)\!\!\!\!\!&=0.1\mathbf{x}_i(k)+\nu_i(k)+ 0.05\omega_i(k),\\
\mathbf{y}_i(k)\!\!\!\!\!&= {c}_i \mathbf{x}_i(k),
\end{array}\right.
\end{align}
where ${c}_i = [c_{i1};\dots;c_{in}]$ with $c_{i(i+1)} = 1$, $c_{ij} = 0$, $\forall i\in [1;n-1], \forall j \neq i+1$, $c_{nn} = 1$, $c_{nj} = 0$, $\forall j \in [1;n-1]$, $\nu_i(k) = 0.145$, $\omega_1(k) = 0$, and $\omega_i(k) = \mathbf{y}_{(i-1)i}(k)$, $\forall i\in [2;n]$. The state set is $\mathbb{X}_i=\mathbb{X}_{i0}= ]0~ 0.6[$, the input set is $\mathbb U_i = \{0.145\}$, the secret set is $\mathbb{X}_{1S} = ]0~0.2]$, $\mathbb{X}_{2S}= [0.4~0.6[$, $\mathbb{X}_{iS}=]0~ 0.6[$, $\forall i\in [3;n]$, 
the output set is $\mathbb Y_{i(i+1)}=]0~ 0.6[$, $\mathbb Y_{ij}=0$, $\forall i\in [1;n-1]$, $\forall j \neq i+1$, $\mathbb Y_{nn} = ]0~0.6[$, $\mathbb Y_{nj} = 0$, $\forall j \in [1;n-1]$. 
and the internal input set is $\mathbb W_i=\prod_{j=1, j\neq i}^{n} \mathbb Y_{ji}$.
One can verify that $\Sigma = \mathcal{I}_{0_n}(\Sigma_1,\dots,\Sigma_n)$. The main goal of this example is to verify approximate initial-state opacity of the concrete network using its finite abstraction. Now, let us construct a finite abstraction of $\Sigma$ compositionally with accuracy $\varepsilon = 0.25$ as defined in \eqref{er}, which preserves initial-state opacity. We apply our main results of previous sections to achieve this goal.

Consider functions $V_i = |x_i- x'_i|$, $\forall i \in [1;n]$. It can be is readily verified that $V_i$ are $\delta$-ISS Lyapunov functions for subsystems $\Sigma_i$ satisfying \eqref{eq:ISTFC1} and \eqref{eq:ISTFC2} in Definition \ref{ass:1}, with $\kappa_i (s) = (1-a_i)s = 0.9s$, $\rho_{exti}(s) = \hat \gamma_i (s) =\underline{\alpha}_i (s)= \overline{\alpha}_i (s)= s$, and $\rho_{inti} (s)= 0.05s$. In addition, the Lipschitz assumption defined in \eqref{lipschitz} holds with $\ell_i (s) = s$. Accordingly, the desired precision for Algorithm \ref{quantialgo} is $\varpi=0.25$. It is seen that the system is made up of $n$ identical subsystems in a cascade interconnection, thus, the resulting directed graph $G = (I,E)$ is specified by $I = [1;n]$, $E = \{(1,2),(2,3),(3,4),\dots,(n-1,n)\}$. Each of the subsystem is a strongly connected component of $G$. The small-gain condition \eqref{smallgaincur} is satisfied readily.
Then, by applying Algorithm \ref{quantialgo}, we obtain proper pairs of local parameters $(\varpi_i,\vartheta_i)= (0.25, 0.25)$ for all of the subsystems. Then, a suitable tuple $q_i = (\eta_i, \mu_i, \theta_i, \phi_i) = (0.2,0,0,0)$ of quantization parameters is chosen such that inequality \eqref{secquantinit} for the abstraction $\hat \Sigma_i$ of each subsystem $\Sigma_i$ is satisfied. Note that the choice of local quantization parameters is suitable in terms of preserving opacity, regardless of the number of subsystems (i.e. $n$). Next, we construct local abstractions for subsystems as in Subsection \ref{absCur}. Using the result in Theorem \ref{thm:2}, one can verify that $V_i = |x_i- x'_i|$ is a $\varpi_i$-InitSOPSF from each $\Sigma_i$ to its abstraction $\hat \Sigma_i$. Furthermore, by the compositionality result in Theorem \ref{thm:3}, we obtain that $\tilde{V} = \max\limits_{i}\{V_i(x_{i},\hat{x}_{i}) \} = \max\limits_{i}\{|x_i- x'_i|\}$ is a  $\varpi$-InitSOPSF from $\Sigma$ to $\hat \Sigma = \mathcal{I}_{0_n}(\hat \Sigma_1,\dots,\hat \Sigma_n)$ satisfying the conditions in Definition \ref{def:SFD2} with  $\varpi = \max\limits_{i} \varpi_i = 0.25$ and $\alpha(s) = \{\max\limits_{i}\{\alpha^{-1}_i(s)\}\}^{-1} = s$.

Now, let us verify opacity of $\Sigma$ using the interconnected abstraction $\hat \Sigma$. Note that given the local quantization parameters $q_i = (\eta_i, \mu_i, \theta_i, \phi_i) = (0.2,0,0,0)$, each local state set $\mathbb{X}_i$ is discretized into $2$ discrete states as $\hat{\mathbb{X}}_i = \{0.2,0.4\}$, which implies that the state space of $\hat \Sigma$ is $\hat{\mathbb{X}}_0 = \hat{\mathbb{X}} = \{0.2,0.4\}^{n}$, the discrete secret state set is $\hat{\mathbb X}_S =  \{0.2\} \times \{0.4\} \times \{0.2,0.4\}^{n-2}$, the discrete input set is $\hat{\mathbb{U}}= \{0.145\}^{n}$, and the output set is $\hat{\mathbb Y} = \{0.2,0.4\}$. In order to check opacity of the interconnected abstraction, we first show an example of a network consisting of $2$ subsystems, as shown in Figure~\ref{exautomata1}. The two smaller automata in the left represent the symbolic subsystems and the one in the right represents the interconnected abstraction for the whole network. Each circle is labeled by the state (top half) and the corresponding output (bottom half). Initial states are distinguished by being the target
of a sourceless arrow. One can easily see that $\mathcal{I}(\hat \Sigma_1,\hat \Sigma_2)$ is $0$-approximate initial-state opaque, since for any run starting from secret state $aA$, there exists a run from non-secret state $AA$ such that the output trajectories are exactly the same. Next, let us see the case when the number of subsystems is $n = 3$. As seen in Figure \ref{exautomata2}, for any run starting from any secret state, i.e. $aAa$ and $aAA$, there exists a run from a non-secret state, i.e. $Aaa$ and $AAA$, such that the output trajectories are exactly the same. Due to lack of space, we do not plot the automata for the case of $n=4$, but we verified that the network is still $0$-approximate initial-state opaque. 
We expect that the interconnected network holds this property regardless of the number of subsystems due to the homogeneity of subsystems and the structure of the network topology. Thus, one can conclude that $\hat \Sigma = \mathcal{I}_{0_n}(\hat \Sigma_1,\dots,\hat \Sigma_n)$ is $0$-approximate initial-state opaque. Therefore, by Proposition \ref{thm:InitSOP}, we obtain that the original network $\Sigma = \mathcal{I}_{0_n}(\Sigma_1,\dots,\Sigma_n)$ is $0.5$-approximate initial-state opaque.

\begin{figure}
\begin{tikzpicture}[->,>=stealth',shorten >=1pt,auto,node distance=2.2cm,inner sep=2pt, initial text =,
every state/.style={draw=black,fill=white,state/.style=state with output},
accepting/.style={draw=black,thick,fill=red!80!green,text=white},bend angle=25]

\node[] at (-2,0) {$\hat \Sigma_1$:};
\node[] at (-2,-2.2) {$\hat\Sigma_2$:};
\node[] at (6,-1.2) {$\mathcal{I}(\hat \Sigma_1,\hat \Sigma_2)$:};
\node[state with output,accepting, initial] (A)                    {$a$ \nodepart{lower} $0y$};
\node[state with output, initial right]         (B) [right of=A] {$A$ \nodepart{lower} $0Y$};
\node[state with output, initial] (C)   [below of=A]                 {$a$ \nodepart{lower} $0y$};
\node[state with output,accepting, initial right]         (D) [right of=C] {$A$ \nodepart{lower} $0Y$};

\path (A) edge [loop above]    node {} (A)

(B) edge              node {} (A)

(C) edge [loop above]    node {a/A} (C)

(D) edge     [loop above]  node {A} (D)
      edge           node {a/A} (C)
;

\node[state with output, initial] (A1)       at (8,0)             {$aa$ \nodepart{lower} $y$};
\node[state with output, initial right]         (B1) [right of=A1] {$Aa$ \nodepart{lower} $y$};
\node[state with output, accepting, initial]         (D1) [below of=A1] {$aA$ \nodepart{lower} $Y$};
\node[state with output, initial right]         (H1) [below of=B1]       {$AA$ \nodepart{lower} $Y$};

\path (A1) edge [loop above]    node {} (A1)

(B1) edge              node {} (A1)
(D1) edge             node {} (A1)

(H1) edge   node {} (A1)
edge          node {} (D1);

\end{tikzpicture}
\caption{Compositional abstraction of an interconnected discrete-time linear system consisting of 2 subsystems.}
\label{exautomata1}
\vspace{-0.25cm}
\end{figure}
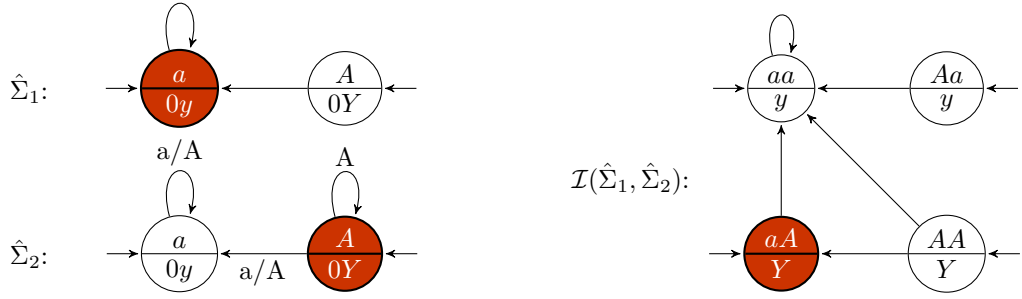

\begin{figure}
	\begin{tikzpicture}[->,>=stealth',shorten >=1pt,auto,node distance=2.3cm, inner sep=2pt, initial text =,
	every state/.style={draw=black,fill=white,state/.style=state with output},
	accepting/.style={draw=black,thick,fill=red!80!green,text=white},bend angle=25]
	
	\node[] at (-2,0) {$\hat \Sigma_1$:};
	\node[] at (-2,-2.3) {$\hat\Sigma_2$:};
	\node[] at (-2,-4.6) {$\hat\Sigma_3$:};
	\node[] at (6,0) {$\mathcal{I}(\hat \Sigma_1,\hat \Sigma_2,\hat\Sigma_3)$:};
	
	\node[state with output,accepting,initial] (A0)                    {$a$ \nodepart{lower} $0y0$};
	\node[state with output,initial right]         (B0) [right of=A0] {$A$ \nodepart{lower} $0Y0$};
	
	\node[state with output,initial] (C0)   [below of=A0]                 {$a$ \nodepart{lower} $00y$};
	\node[state with output,accepting,initial right]         (D0) [right of=C0] {$A$ \nodepart{lower} $00Y$};

	\node[state with output,initial] (E0)   [below of=C0]                 {$a$ \nodepart{lower} $00y$};
	\node[state with output,initial right]         (F0) [right of=E0] {$A$ \nodepart{lower} $00Y$};

	\path (A0) edge [loop above]    node {} (A0)
	
	(B0) edge              node {} (A0)
	
	(C0) edge [loop above]    node {a/A} (C0)
	
	(D0) edge     [loop above]  node {A} (D0)
	edge           node {a/A} (C0)
	
	(E0) edge [loop above]    node {a/A} (E0)
	
	(F0) edge     [loop above]  node {A} (F0)
	edge           node {a/A} (E0)

	;

	\node[state with output,initial above] (A)      at (8,-1.5)              {$aaa$ \nodepart{lower} $y$};
	\node[state with output, accepting,initial below]         (B) [right of=A] {$aAa$ \nodepart{lower} $y$};
	\node[state with output,initial]         (D) [left of=A] {$Aaa$ \nodepart{lower} $y$};
	\node[state with output,initial right]         (E) [right of=B] {$AAa$ \nodepart{lower} $y$};
	\node[state with output,initial below]         (H) [below of=A]       {$aaA$ \nodepart{lower} $Y$};
	\node[state with output,initial]         (F) [left of=H] {$AAA$ \nodepart{lower} $Y$};
	\node[state with output, accepting,initial below]         (G) [right of=H] {$aAA$ \nodepart{lower} $Y$};
	\node[state with output,initial right]         (I) [right of=G]       {$AaA$ \nodepart{lower} $Y$};

	\path (A) edge  [loop left]    node {} (A)
	(B) edge              node {} (A)
	(D) edge      [bend right]        node {} (A)
	(E) edge          node {} (B)
	edge    [bend right]      node {} (A)
	(F) edge          node {} (H)
	edge           node {} (A)
	
	(G) edge          node {} (H)
	edge          node {} (A)
	(H) edge    node {} (A)
	(I) edge    node {} (A);
	\end{tikzpicture}
	\caption{Compositional abstraction of an interconnected discrete-time linear system consisting of 3 subsystems.}
	\label{exautomata2}
\end{figure}
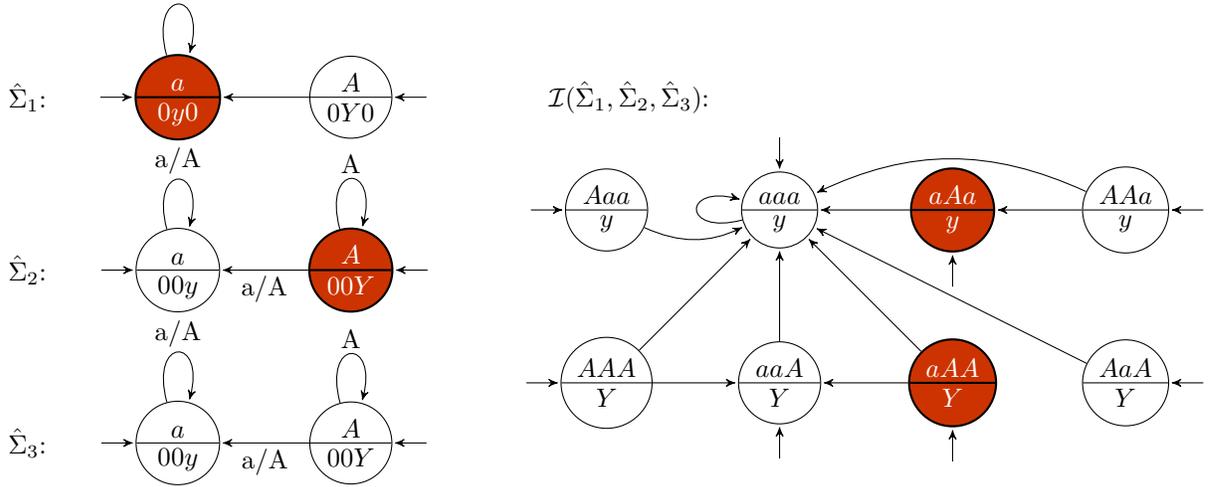

\begin{remark}
	In Figure~\ref{exautomata1}, for simplicity of demonstration, we use symbol ``$a$" (resp. ``$A$") to show state $x_i = 0.2$ (resp. $x_i = 0.4$). The symbols ``$aa$", ``$Aa$", ``$aA$" and ``$AA$" represent the state vectors $x = [0.2;0.2]$, $x = [0.4;0.2]$, $x = [0.2;0.4]$, and $x = [0.4;0.4]$, respectively.
	The lower parts of the states indicate the outputs of the states, where the symbols ``$y$" and ``$Y$" represent, respectively, the output $y_{ij}=0.2$ and $y_{ij}=0.4$. Similarly, symbols ``$ 0y$" and ``$ 0Y$" represent the output vectors ${y} = [0;0.2]$ and ${y} = [0;0.4]$. The states marked in red represent the secret states.
	The symbols on the edges show the internal input coming from other subsystems. The symbols used in Figure~\ref{exautomata2} represent similar meanings.
\end{remark}

\section{Conclusion}
In this paper, we proposed a methodology to compositionally construct opacity-preserving symbolic models of interconnected discrete-time control systems. New notions of so-called opacity-preserving simulation functions are introduced to characterize the relations between two systems in terms of preservation of opacity. By leveraging these simulation functions, we constructed  abstractions of the subsystems, while preserving the opacity properties. Then, a symbolic model of the network can be obtained by interconnecting the local finite abstractions while retaining the opacity property. An illustrative example was presented to describe the design procedure of the local quantization parameters. Finally, we applied our main results to a linear interconnected system. 

\bibliographystyle{IEEEtran}      
\bibliography{biblio} 	
\end{document}